\documentclass[thmsa]{article}
\usepackage{amsfonts}
\usepackage{amsmath}
\usepackage{amssymb}
\usepackage{youngtab}
\usepackage{amsfonts}
\usepackage{amsmath}
\usepackage{epsfig,multicol}
\usepackage{graphicx}

\input{tcilatex}

\begin{document}

\setcounter{page}{0} \topmargin0pt \oddsidemargin5mm \renewcommand{%
\thefootnote}{\fnsymbol{footnote}} \newpage \setcounter{page}{0} 
\begin{titlepage}
\begin{flushright}
\end{flushright}
\vspace{0.5cm}
\begin{center}
{\Large {\bf PT Symmetry of the non-Hermitian XX Spin-Chain: Non-local Bulk Interaction from Complex Boundary Fields}}

\vspace{0.8cm}
{ \large Christian Korff}

\vspace{0.5cm}
{\em
Department of Mathematics, University of Glasgow, \\
University Gardens, Glasgow G12 8QW, UK}
\end{center}
\vspace{0.2cm}

\renewcommand{\thefootnote}{\arabic{footnote}}
\setcounter{footnote}{0}

\begin{abstract}
The XX spin-chain with non-Hermitian diagonal boundary conditions is shown to be quasi-Hermitian
for special values of the boundary parameters. This is proved by explicit construction of a new inner
product employing a "quasi-fermion" algebra in momentum space where creation and annihilation
operators are not related via Hermitian conjugation. For a special example, when the boundary fields
lie on the imaginary axis, we show the spectral equivalence of the quasi-Hermitian XX spin-chain with
a non-local fermion model, where long range hopping of the particles occurs as the non-Hermitian
boundary fields increase in strength. The corresponding Hamiltonian interpolates between the open
XX and the quantum group invariant XXZ model at the free fermion point. For an even number of
sites the former is known to be related to a CFT with central charge $c=1$, while the latter has been
connected to a logarithmic CFT with central charge $c=-2$. We discuss the underlying algebraic
structures and show that for an odd number of sites the superalgebra symmetry
$U(\mathfrak{gl}(1|1))$ can be extended from the unit circle along the imaginary axis. We relate the vanishing of
one of its central elements to the appearance of Jordan blocks in the Hamiltonian.
\medskip
\par\noindent
\end{abstract}
\vfill{ \hspace*{-9mm}
\begin{tabular}{l}
\rule{6 cm}{0.05 mm}\\
c.korff@maths.gla.ac.uk
\end{tabular}}
\end{titlepage}\newpage

\section{Introduction}

Recent years have seen increasing interest in non-Hermitian quantum
Hamiltonians $H$\ and how to give them a physical sound interpretation. In
this article we revisit an exactly-solvable, one-dimensional, discrete
system, the XX-Hamiltonian with non-Hermitian diagonal boundary conditions,
i.e. 
\begin{equation}
H=\frac{1}{2}\sum_{m=1}^{M-1}\left[ \sigma _{m}^{x}\sigma _{m+1}^{x}+\sigma
_{m}^{y}\sigma _{m+1}^{y}\right] +\frac{\alpha \sigma _{1}^{z}+\beta \sigma
_{M}^{z}}{2}~,\qquad \alpha ,\beta \in \mathbb{C}~.  \label{H}
\end{equation}%
Here $\sigma _{m}^{x,y,z}$ denote the Pauli matrices acting in the $m^{\text{%
th}}$ copy of an $M$-fold tensor product $V^{\otimes M}$ of a
two-dimensional complex vector space $V=\mathbb{C}v_{+}\oplus \mathbb{C}%
v_{-} $. For real values the parameters $\alpha ,\beta \in \mathbb{R}$ have
the physical interpretation of external magnetic fields located at the
boundary. Here we are interested in investigating the case of \emph{complex}
fields, $\alpha ,\beta \in \mathbb{C}$. Another variant of the
XX-Hamiltonian we shall consider is%
\begin{equation}
H^{\prime }=\frac{1}{2}\sum_{m=1}^{M-1}\left[ \sigma _{m}^{x}\sigma
_{m+1}^{x}+\sigma _{m}^{y}\sigma _{m+1}^{y}-\beta ~\sigma _{m}^{z}-\alpha
~\sigma _{m+1}^{z}\right] =H-\frac{\alpha +\beta }{2}~\sum_{m=1}^{M}\sigma
_{m}^{z}  \label{Hprime}
\end{equation}%
which up to boundary terms coincides with the closed or periodic XX
spin-chain Hamiltonian in an external magnetic field. Both Hamiltonians can
be diagonalized via a Jordan-Wigner transformation for arbitrary complex
values of the boundary parameters $\alpha ,\beta $. However, the
Hamiltonians $H,H^{\prime }$ with $\alpha ,\beta \in \mathbb{C}$ are in
general non-Hermitian, 
\begin{equation}
H\neq H^{\ast }=H(\bar{\alpha},\bar{\beta}),  \label{Hcross}
\end{equation}%
and one needs to explain how a meaningful quantum mechanical system in terms
of $H,~H^{\prime }$ can be defined. On physical grounds the time evolution
operator $U(t)=\exp (itH),$ $t>0$ ought to be unitary and, hence, $H$ needs
to be Hermitian. Thus, a non-Hermitian Hamiltonian, $H\neq H^{\ast },$
appears at first sight to be in contradiction with conventional quantum
mechanics. However, under certain assumption the Hamiltonian might turn out
to be quasi-Hermitian \cite{SGH92, Most04}. That is, there exists a \emph{%
positive}, \emph{Hermitian} and \emph{invertible} operator $\eta $ satisfying%
\begin{equation}
\eta H=H^{\ast }\eta \ .  \label{eta0}
\end{equation}%
This allows one to either introduce a new inner product,%
\begin{equation}
\langle v,w\rangle _{\eta }:=\langle v,\eta w\rangle ,\qquad v,w\in 
\mathfrak{H}  \label{etaprod0}
\end{equation}%
on the state space $\mathfrak{H}$ with respect to which $H$ becomes
Hermitian or perform a similarity transformation to a new, Hermitian
Hamiltonian with respect to the original inner product,%
\begin{equation}
h=\eta ^{1/2}H\eta ^{-1/2}\ .  \label{h0}
\end{equation}

Given a non-Hermitian Hamiltonian it is in general rather difficult to
determine whether it is quasi-Hermitian. A special subclass of
quasi-Hermitian systems where this turns out to be easier are those
distinguished by PT-symmetry, i.e. Hamiltonian systems where the
eigenvectors can be chosen such that they are eigenvectors under a joint
parity and time reversal transformation; see e.g. \cite{Bender07} and
references therein. In this case the quasi-Hermiticity operator $\eta $
enjoys further constraints which we will discuss in the text below.
Employing the exact solvability of the non-Hermitian XX spin-chain we will
establish for which values of the boundary parameters it is quasi-Hermitian
and show its PT-symmetry.

\subsection{Long range bulk interaction from non-Hermitian boundary fields}

The above non-Hermitian Hamiltonians are of special interest because of
their underlying algebraic structure and we will show that it plays an
important role in the interpretation of the new inner product. Setting $%
\alpha =-\beta =\sqrt{-1}$ the Hamilton (\ref{H}) is the $U_{q}(sl_{2})$%
-invariant XXZ spin-chain evaluated at the free fermion point $q=\sqrt{-1}$ 
\cite{Alc87, PS90}. In the thermodynamic limit it has been suggested that
the system is closely connected with critical dense polymers effectively
described by a logarithmic conformal field theory with central charge $c=-2$ 
\cite{Saleur92, PRZ06,RS07}. Moreover, it has been stressed that such models
should be described in terms of non-local degrees of freedom, such as
\textquotedblleft connectivities\textquotedblright , see e.g. \cite{PRZ06}.

On the other hand, setting $\alpha =\beta =0$ the chain (\ref{H}) becomes
Hermitian and is related to the Ashkin-Teller model, see \cite{Alc87}. In
the thermodynamic limit the system is now described by an ordinary conformal
field theory with central charge $c=1$.

Thus, the non-Hermitian Hamiltonian%
\begin{equation}
H_{g}=\frac{1}{2}\sum_{m=1}^{M-1}\left[ \sigma _{m}^{x}\sigma
_{m+1}^{x}+\sigma _{m}^{y}\sigma _{m+1}^{y}+ig(\sigma _{m}^{z}-\sigma
_{m+1}^{z})\right] =H_{-g}^{\ast },\qquad 0<g<1,  \label{Hg}
\end{equation}%
interpolates between these two special cases. We will show for small values
of the coupling constant $g$ its spectral equivalence with the Hermitian
Hamiltonian%
\begin{equation}
h_{g^{2}}=-\sum_{n>0}\sum_{x=1}^{M-n}p_{x}^{(n)}(g^{2})\left[ c_{x}^{\ast
}c_{x+n}-c_{x}c_{x+n}^{\ast }\right] ,  \label{hg}
\end{equation}%
where $c_{x}^{\ast },c_{x}$ are fermionic creation and annihilation
operators at lattice site $x$. The hopping probability between a site $x$
and its $n^{\text{th}}$ neighbour is encoded in the real coefficients $%
p_{x}^{(n)}$ which only depend on $g^{2}$ and vanish for $n$ even, $%
p_{x}^{(2n)}=0$. At $g=0$ we only have nearest neighbour hopping,%
\begin{equation}
h_{0}=-\sum_{x=1}^{M-1}(c_{x}^{\ast }c_{x+1}-c_{x}c_{x+1}^{\ast })\;.
\end{equation}%
Our perturbative calculation will show that as $g$ increases the bulk
interaction becomes more and more long-range by successively
\textquotedblleft switching on\textquotedblright the various coefficients $%
p_{x}^{(2n+1)}$ starting from the boundary sites $x=1,M$. More precisely, we
find up to order $g^{8}$ that the nonvanishing contributions are%
\begin{eqnarray}
p_{x}^{(1)} &=&1-\frac{128g^{2}+8g^{4}+g^{6}}{512}(\delta _{x,1}+\delta
_{x,M-1})  \notag \\
&&-\frac{8g^{4}+3g^{6}}{512}(\delta _{x,2}+\delta _{x,M-2})+\frac{g^{6}}{256}%
(\delta _{x,3}+\delta _{x,M-3})+O(g^{8})  \notag \\
p_{x}^{(3)} &=&\frac{20g^{4}+3g^{6}}{256}~(\delta _{x,1}+\delta _{x,M-3})+%
\frac{5g^{6}}{512}(\delta _{x,2}+\delta _{x,M-4})+O(g^{8})  \notag \\
p_{x}^{(5)} &=&-\frac{23g^{6}}{512}(\delta _{x,1}+\delta _{x,M-5})+O(g^{8})
\end{eqnarray}%
Figure 1 indicates which hopping amplitudes, i.e. the probabilities for a
fermion to jump between two lattice sites, are modified up to order $g^{8}$.
We have only depicted the additional contributions for $g>0$: whenever an
arc connects two lattice sites $x$ and $x+n$ then there is a non-vanishing
coefficient $p_{x}^{(n)}$.\newpage 
\begin{equation*}
\includegraphics[scale=0.9]{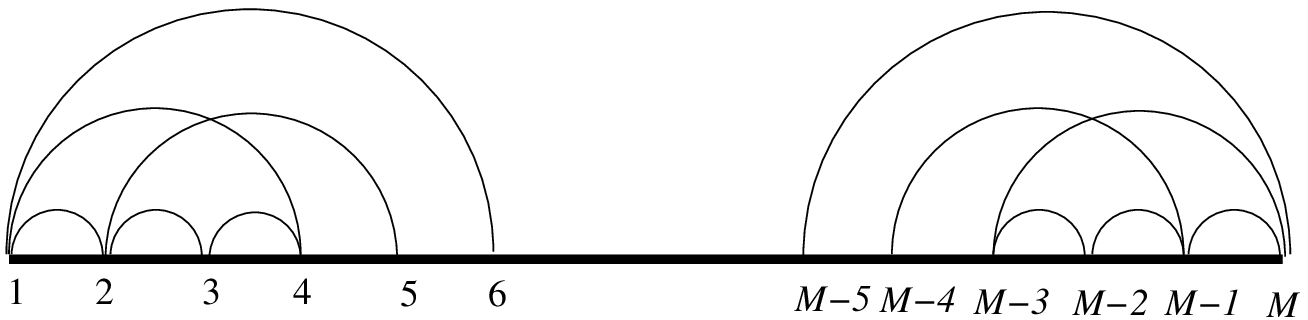}.
\end{equation*}%
{\small \medskip }

\begin{center}
{\small Figure 1. Graphical depiction of the long range hopping.\medskip }
\end{center}

\noindent Both variants (\ref{Hg}) and (\ref{hg}) of the Hamiltonian have
their advantages:

\begin{itemize}
\item The Hermitian Hamiltonian $h$ serves physical intuition. It can be
directly interpreted and clearly shows the long range nature of the
interaction which is not apparent in (\ref{Hg}). It suggests that the long
range nature of the interaction with increasing $g$ is behind the singular
change from an \textquotedblleft ordinary\textquotedblright\ conformal field
theory with $c=1$ for $0\leq g<1$ to a logarithmic one with $c=-2$ at $g=1$.

\item In contrast, the variant (\ref{Hg}) highlights the algebraic
properties, such as integrability and quantum group invariance. The $PT$%
-symmetry and quasi-Hermiticity of the quantum group invariant $XXZ$
spin-chain for higher roots of unity $q=\exp (i\pi /r),$ $r>2$ have been
discussed in another paper \cite{KW07}.
\end{itemize}

Note that the appearance of a long range bulk interaction does not
contradict physical intuition. Naively one would argue that boundary terms
will become unimportant in the thermodynamic limit when the number of sites
tends to infinity. However, the applicability of this statement crucially
depends on the nature of the boundary conditions, see for instance \cite%
{KZJ00}. In the present article we consider a different, novel example of
this phenomenon. The complex boundary terms render the Hamiltonian
non-Hermitian, hence they imply a drastic non-local change in the
mathematical structure of the state space: the introduction of a new inner
product. We will see that this is the case by observing that the similarity
transformation which maps $H_{g}$ into $h_{g^{2}}$ is highly
non-local.\smallskip

\subsection{Outline of the article and summary of results}

The purpose of the present article is to relate the notions of $PT$-symmetry
and quasi-Hermiticity to another simple but non-trivial model. While the $XX$
spin-chain with non-Hermitian boundary terms has been discussed previously
in the literature (see e.g. \cite{HR92, AR93, BW97,BW99, Bil00}) it has not
been investigated for which parameter values $\alpha ,\beta $ the
Hamiltonian is quasi-Hermitian and what the corresponding Hermitian systems
are. We will show that for $\alpha =\bar{\beta}$ inside the unit disc the
Hamiltonians (\ref{H}), (\ref{Hprime}) are quasi-Hermitian, this condition
on $\alpha ,\beta $ ensures also $PT$-symmetry. We then will discuss for
which values the mentioned Hamiltonians cease to be quasi-Hermitian and when
their spectrum contains complex eigenvalues. The special case $\alpha =\bar{%
\beta}$ of the non-Hermitian XX spin-chain, which includes (\ref{Hg}), has
not been discussed in detail previously.\smallskip

For the benefit of the reader we summarize the new results contained in this
article:

\begin{itemize}
\item After a brief review of the notions of quasi-Hermiticity and
PT-symmetry on the lattice in Section 2, we establish in Section 3 that the
spectra of the Hamiltonians (\ref{H}) and (\ref{Hprime}) are real for
boundary parameters within the unit disc, $\alpha =\bar{\beta}\in \mathbb{C}$
and $|\alpha |<1$. Performing a Jordan-Wigner transformation the XX
spin-chain can be reformulated as a non-trivial fermion model and its
spectrum can be described in terms of quasi-particle excitations in momentum
space. We will give an elementary proof that all quasi-momenta (Bethe roots)
lie on the unit circle and state for $\alpha =-\beta =ig,\;0<g<1$ a set of
palindromic polynomials whose roots give the quasi-momenta $k$ and the
associated energies $2\cos k$.

\item In Section 4 we highlight that unlike in the case of real boundary
fields $\alpha ,\beta \in \mathbb{R}$, the Jordan-Wigner transformation does 
\emph{not} lead to a well-defined fermion algebra in momentum space, but
instead one has two sets $\{\hat{c}_{k}^{\ast },\hat{c}_{k}\}$ and $\{\hat{d}%
_{k}^{\ast },\hat{d}_{k}\}$ of creation and annihilation operators which
satisfy the relations%
\begin{equation}
\lbrack \hat{c}_{k}^{\ast },\hat{d}_{k^{\prime }}]_{+}=\delta _{k,k^{\prime
}},\qquad \lbrack \hat{c}_{k},\hat{d}_{k^{\prime }}]_{+}=[\hat{c}_{k}^{\ast
},\hat{d}_{k^{\prime }}^{\ast }]_{+}=0,\qquad \hat{c}_{k}^{\ast }\neq \hat{d}%
_{k}\;.
\end{equation}%
Here $k=-i\ln z\in \mathbb{R}$ is a quasi-momentum, $[A,B]_{+}=AB+BA$ is the
anti-commutator and $\ast $ denotes the Hermitian adjoint with respect to
the original inner product, where the Hamiltonian is non-Hermitian.
Employing the above operators we explicitly construct an $\eta $ which not
only renders the Hamiltonian Hermitian but obeys the more restrictive
condition%
\begin{equation}
\eta \hat{c}_{k}^{\ast }=\hat{d}_{k}^{\ast }\eta \qquad \text{and}\qquad
\eta \hat{d}_{k}=\hat{c}_{k}\eta \ .
\end{equation}%
That is, with respect to the new inner product (\ref{etaprod0}) $\{\hat{c}%
_{k}^{\ast },\hat{d}_{k}\}$ satisfy the canonical anti-commutation relations 
\emph{and} are the Hermitian adjoint of each other.

\item Also in Section 4 we extend for $\alpha =-\beta =ig,\;0\leq g\leq 1$
the well known $U_{q=i}(sl_{2})$-symmetry of (\ref{Hg}) from the unit
circle, $g=1$, into the unit disc, $0\leq g<1$ and discuss the invariance of
the Hamiltonian for odd and even number of sites. Both cases show remarkable
differences. The quantum group symmetry is closely connected with so-called
fermionic zero modes and we show that at the coupling values $g$ where the
anticommutator of the associated fermionic creation and annihilation
operators vanishes the Hamiltonian possesses non-trivial Jordan blocks.
Thus, the Hamiltonian ceases to be quasi-Hermitian. For even numbers of
sites this happens at $g=1$, i.e. on the unit circle, as was observed
previously \cite{PRZ06}. Here we show that it also happens for odd numbers
of sites albeit at the value $g_{\text{max}}=\sqrt{(M+1)/(M-1)}$ which
approaches the value one in the limit $M\rightarrow \infty $. It also
signals the onset of complex eigenvalues for values of $g>g_{\max }$, i.e.
outside the unit disc.

\item In Section 5 we present a perturbative calculation of $\eta $ in terms
of the coupling parameter $0<g\ll 1$ in (\ref{Hg}). While we follow closely
the steps previously put forward in the literature \cite{BBJ04,CFMFAF06}, we
give a novel derivation of the coefficients in the perturbation series
expansion. Using these results we obtain the expression (\ref{hg}). We also
will derive some closed expressions for small number of sites $M=3,4$ and 5
in Section 6. For $g=1$ and odd lattice sites we formulate a conjecture on
an alternative way of computing the new inner product using the
Temperley-Lieb algebra.

\item In Section 7 we will investigate briefly the case of general complex
boundary fields inside the unit disc, i.e. when $\alpha =\bar{\beta}%
=ge^{i\theta }$ with $\theta \neq \pi /2$ and real in order to make contact
with the discussion in \cite{HR92}.

\item Section 8 contains the conclusions.
\end{itemize}

\section{Quasi-Hermiticity and PT-invariance}

The state space of the non-Hermitian systems (\ref{H}) and (\ref{Hprime}) is
a spin-chain of $M$ sites represented in terms of the tensor product $%
V^{\otimes M}$ of the two-dimensional complex vector space $V$ with
orthonormal basis $\{v_{\pm }\}$ such that $\sigma ^{z}v_{\pm }=\pm v_{\pm }$%
. The tensor product is then spanned by the vectors%
\begin{equation}
\{\left\vert \varepsilon _{1},...,\varepsilon _{M}\right\rangle \equiv
v_{\varepsilon _{1}}\otimes \cdots \otimes v_{\varepsilon
_{M}}:\;\varepsilon _{m}=\pm 1\}  \label{spinbasis}
\end{equation}%
This particular choice of basis vectors is motivated by the axial symmetry
of the Hamiltonians $H,$~$H^{\prime }$ which commute with the total spin
operator,%
\begin{equation}
S^{z}=\frac{1}{2}\tsum_{m=1}^{M}\sigma _{m}^{z},\qquad \lbrack
H,S^{z}]=[H^{\prime },S^{z}]=0\ .  \label{Sz}
\end{equation}%
The matrices (\ref{H}) and (\ref{Hprime}) are defined with respect to the
inner product%
\begin{equation}
\langle \varepsilon _{1},...,\varepsilon _{M}\left\vert \varepsilon
_{1}^{\prime },...,\varepsilon _{M}^{\prime }\right\rangle
=\prod_{i=1}^{M}\langle v_{\varepsilon _{i}},v_{\varepsilon _{i}^{\prime
}}\rangle ,\qquad \left\langle v_{\varepsilon },v_{\varepsilon ^{\prime
}}\right\rangle =\delta _{\varepsilon \varepsilon ^{\prime }}\;.
\label{origprod}
\end{equation}

As already pointed out in the introduction both Hamiltonians, $H$ and $%
H^{\prime }$, are in general non-Hermitian for arbitrary complex boundary
parameters $\alpha ,\beta $. However, in the above basis (\ref{spinbasis})
both of them are symmetric,%
\begin{equation*}
H=H^{t}\qquad \text{and}\qquad H^{\prime }=H^{\prime t}\;.
\end{equation*}%
We now wish to determine for which values there exists a map $\eta
:V^{\otimes M}\rightarrow V^{\otimes M}$ possessing the following properties:

\begin{enumerate}
\item $\eta $ is Hermitian, $\eta =\eta ^{\ast }$, invertible, $\det \eta >0$%
, and positive definite, $\eta >0$.

\item $\eta $ intertwines the Hamiltonian $H$ (respectively $H^{\prime }$)
with its Hermitian adjoint,%
\begin{equation}
\eta H=H^{\ast }\eta \;.  \label{etainter}
\end{equation}
\end{enumerate}

Because of the axial symmetry present in $H$ and $H^{\prime }$ we add the
following requirement which implies that $\eta $ intertwines $H$ and $%
H^{\prime }$ at the same time,%
\begin{equation}
\lbrack \eta ,S^{z}]=0\;.  \label{etaSz}
\end{equation}%
Provided that such a map $\eta $ exists we define a new inner product and
with it a new Hilbert space structure via%
\begin{equation}
\left\langle v,w\right\rangle _{\eta }:=\left\langle v,\eta w\right\rangle
,\qquad v,w\in V^{\otimes M}\ .  \label{etaprod}
\end{equation}%
Using the intertwining property (2) it is obvious that the Hamiltonian is
Hermitian with respect to the new inner product. The properties listed under
(1) ensure that the new inner product is well-defined.

There is an alternative to introducing a new Hilbert space structure. Since $%
\eta >0$ it follows that there exists a unique positive definite square root 
$\eta ^{\frac{1}{2}}>0$ which allows one to define the Hermitian Hamiltonian 
\begin{equation}
h=\eta ^{\frac{1}{2}}H\eta ^{-\frac{1}{2}}  \label{h}
\end{equation}%
with respect to the original inner product.

Both approaches have their advantages and disadvantages. At first glance one
might prefer to work in the original Hilbert space with the
\textquotedblleft gauge transformed\textquotedblright\; Hamiltonian $h$, as
this allows for a direct physical interpretation without having to insert
another operator when taking scalar products. However, for practical
purposes it is often more feasible to work with $H$ and $\eta $. Even if it
is possible to construct $\eta $ explicitly, the computation of its square
root and, thus, the calculation of $h$ is another technically complicated
step. We will perform this computation for (\ref{Hg}) to low orders in $g$
employing perturbation theory in Section 5. In contrast, the original
Hamiltonian $H$ has usually a simpler form which allows one to read off
certain symmetries. The XX-chain with imaginary boundary fields is a
concrete example.

\subsection{Parity, Time and Spin Reversal on the lattice}

On the set of basis vectors (\ref{spinbasis}) we define the parity reversal
operator $P$ by%
\begin{equation}
P\left\vert \varepsilon _{1},...,\varepsilon _{M}\right\rangle =\left\vert
\varepsilon _{M},\varepsilon _{M-1},...,\varepsilon _{1}\right\rangle
\label{P}
\end{equation}%
and extend its action to the whole space by linearity. From this definition
it is immediate to see that%
\begin{equation}
PHP=H(\beta ,\alpha )\ .  \label{PHP}
\end{equation}%
The time reversal operator $T$ acts on the basis vectors as identity, 
\begin{equation}
T\left\vert \varepsilon _{1},...,\varepsilon _{M}\right\rangle =\left\vert
\varepsilon _{1},...,\varepsilon _{M}\right\rangle ,  \label{T}
\end{equation}%
but is defined to be \emph{antilinear}, whence any matrix (such as the
Hamiltonian) is transformed into its complex conjugate under the adjoint
action of $T$,%
\begin{equation}
THT=\bar{H}=H(\bar{\alpha},\bar{\beta})=H^{\ast }\ .  \label{THT}
\end{equation}%
Below we will define discrete wave functions which will be transformed into
their complex conjugates under the operator $T$. This justifies the
identification of $T$ with time-reversal.

Upon imposing the constraint%
\begin{equation}
\alpha =\bar{\beta}  \label{con1}
\end{equation}%
the Hamiltonian turns out to be $PT$-invariant, 
\begin{equation}
\alpha =\bar{\beta}:\qquad \lbrack PT,H]=0\ .  \label{PT}
\end{equation}

Note, however, that $PT$-invariance of the Hamiltonian is not a sufficient
criterion to ensure real eigenvalues, this only follows once it is
established that the eigenvectors of the Hamiltonian can be chosen to be
simultaneously eigenvectors of the $PT$-operator. This is not automatically
implied by $PT$-invariance of the Hamiltonian as time reversal $T$ is an
antilinear operator; see e.g. \cite{Weig03}.

For later purposes we also discuss the behaviour under the spin-reversal
operator,%
\begin{equation}
RH(\alpha ,\beta )R=H(-\alpha ,-\beta ),\qquad R=\tprod_{m=1}^{M}\sigma
_{m}^{x}\ .  \label{RHR}
\end{equation}%
Thus, in general spin-reversal symmetry is broken in the presence of
boundary fields.

Provided that the Hamiltonian is not only $PT$-invariant but satisfies the
slightly more restrictive constraint%
\begin{equation}
PHP=THT=H^{\ast },  \label{PTH}
\end{equation}%
as it is the case here for $\alpha =\bar{\beta}$, it is natural to impose
further constraints on $\eta $. Namely, we wish to have%
\begin{equation}
P\eta P=T\eta T=\eta ^{-1}\;.  \label{PTeta}
\end{equation}%
These conditions are compatible with the intertwining property, e.g.%
\begin{equation*}
P\eta HP=P\eta P~H^{\ast }=H~P\eta P=PH^{\ast }\eta P\;.
\end{equation*}%
In essence we are demanding that parity-reversal gives the same $\eta $ up
to inversion. The same applies for time-reversal.

Often it is beneficial to define an additional operator $C$ introduced by
Bender and collaborators (for references see \cite{Bender07}) by setting%
\begin{equation}
C:=P\eta \;.  \label{Cdef}
\end{equation}%
This operator might turn out to have a simpler expression than $\eta $
itself. The aforementioned properties of $\eta $ imply the identities%
\begin{equation}
C^{2}=1,\qquad \lbrack PT,C]=0,\qquad \lbrack H,C]=0\ .  \label{C}
\end{equation}%
The C-operator turned out to have an elegant algebraic expression for the
quantum group invariant XXZ spin-chain; see \cite{KW07}.

\section{Spectrum and eigenvectors of the Hamiltonian}

As mentioned in the introduction the Hamiltonians $H,~H^{\prime }$ can be
diagonalised in terms of free fermions. Using the well known Jordan-Wigner
identities%
\begin{equation}
c_{x}=\left( \tprod_{y<x}\sigma _{y}^{z}\right) \sigma _{x}^{-},\qquad
c_{x}^{\ast }=\left( \tprod_{y<x}\sigma _{y}^{z}\right) \sigma
_{x}^{+},\qquad \text{and}\qquad n_{x}=c_{x}^{\dagger }c_{x}=\frac{1+\sigma
_{x}^{z}}{2}
\end{equation}%
we introduce fermion creation and annihilation operators in
\textquotedblleft position space\textquotedblright\ satisfying the canonical
anti-commutation relations (CAR),%
\begin{equation}
\lbrack c_{x},c_{y}]_{+}=[c_{x}^{\ast },c_{y}^{\ast }]_{+}=0\qquad \text{and}%
\qquad \lbrack c_{x}^{\ast },c_{y}]_{+}=\delta _{x,y},\quad \quad
x,y=1,2,...,M\ .  \label{CAR}
\end{equation}%
For later purposes we also state the transformation properties under parity,
time and spin-reversal,%
\begin{equation}
Pc_{x}=c_{M+1-x}P,\qquad Tc_{x}=c_{x}T,\qquad Rc_{x}=(-1)^{x+1}c_{x}^{\ast
}R\ .  \label{PTRc}
\end{equation}%
In terms of the fermion algebra the Hamiltonians can be rewritten as%
\begin{equation}
H=-\sum_{x=1}^{M-1}\left[ c_{x}^{\ast }c_{x+1}-c_{x}c_{x+1}^{\ast }\right]
+\alpha ~n_{1}+\beta ~n_{M}-\frac{\alpha +\beta }{2},
\end{equation}%
and%
\begin{equation}
H^{\prime }=-\sum_{x=1}^{M-1}\left[ c_{x}^{\ast }c_{x+1}-c_{x}c_{x+1}^{\ast
}+\beta ~n_{x}+\alpha ~n_{x+1}-\frac{\alpha +\beta }{2}\right] \ .
\end{equation}%
Below we shall constrain the boundary parameters $\alpha ,\beta $. For the
moment we leave them arbitrary. We now introduce a \textquotedblleft
discrete wave function\textquotedblright\ $\psi _{z}$ depending on a complex
parameter $z\in \mathbb{C}$ by defining%
\begin{equation}
\hat{c}_{z}^{\ast }=\sum_{x=1}^{M}\psi _{z}(x;\alpha ,\beta )c_{x}^{\ast },
\label{ccross}
\end{equation}%
with 
\begin{equation}
\psi _{z}(x;\alpha ,\beta )=z^{x}-A(z;\alpha ,\beta )~z^{-x}\ .  \label{psi}
\end{equation}%
This ansatz is physically motivated: it is a superposition of two (discrete)
plane waves, one incoming and one reflected. Here $z=\exp (ik)$ with $k$
being the quasi-momentum and $A$ is the reflection coefficient. Note that we
allow here for complex momenta $k$.

Employing the canonical anticommutation relations one easily finds that%
\begin{eqnarray*}
c_{x}^{\ast }H &=&Hc_{x}^{\ast }+c_{x-1}^{\ast }+c_{x+1}^{\ast },\qquad
1<x<M, \\
c_{1}^{\ast }H &=&Hc_{1}^{\ast }-\alpha c_{1}^{\ast }+c_{2}^{\ast }, \\
c_{M}^{\ast }H &=&Hc_{M}^{\ast }+c_{M-1}^{\ast }-\beta c_{M}^{\ast }\ .
\end{eqnarray*}%
Therefore, one has%
\begin{equation}
\lbrack H,\hat{c}_{z}^{\ast }]=-(z+z^{-1})\hat{c}_{z}^{\ast }  \label{specH}
\end{equation}%
provided the coefficient $A$ in the wave function (\ref{psi}) obeys the
identities \cite{Alc87}%
\begin{equation}
A=\frac{1+\alpha z}{1+\alpha /z}=z^{2M}\frac{\beta +z}{\beta +z^{-1}}\ .
\label{bae}
\end{equation}%
As both equations have to hold simultaneously this imposes a constraint on
the allowed values for the parameter $z$ which are specified as the roots of
a polynomial equation of order $2M+2$. From (\ref{specH}) we infer that the
spectrum of the Hamiltonian is composed of quasi-particle excitations with
energy $\varepsilon =z+z^{-1}$ by successively acting with (\ref{ccross}) on
the pseudo-vacuum vector $\left\vert 0\right\rangle =v_{-1}\otimes \cdots
\otimes v_{-1}$,%
\begin{equation}
H\left\vert z_{1},...,z_{l}\right\rangle =\left( -\frac{\alpha +\beta }{2}%
-\tsum_{i=1}^{l}\left( z_{i}+z_{i}^{-1}\right) \right) \left\vert
z_{1},...,z_{l}\right\rangle ,  \label{specH3}
\end{equation}%
where%
\begin{equation}
\left\vert z_{1},...,z_{l}\right\rangle =\hat{c}_{z_{1}}^{\ast }\cdots \hat{c%
}_{z_{l}}^{\ast }\left\vert 0\right\rangle \;,\;.  \label{particlebasis}
\end{equation}%
Thus, it depends on the nature of the solutions $z$ of the equation (\ref%
{bae}) whether the spectrum of the Hamiltonian is real.

\begin{proposition}
Let $\alpha =\bar{\beta}$ in (\ref{H}) and assume that $\alpha $ is inside
the closed unit disc, i.e. $|\alpha |\leq 1$. Then the solutions of (\ref%
{bae}) all lie on the unit circle, hence the quasi-momenta $k_{j}=-i\ln
z_{j} $ are real.
\end{proposition}

\begin{proof}
Let us rewrite the Bethe ansatz equations (\ref{bae}) as%
\begin{equation}
z^{M}~\frac{z+\alpha }{z\bar{\alpha}+1}=z^{-M}~\frac{z^{-1}+\alpha }{z^{-1}%
\bar{\alpha}+1}~.  \label{bae2}
\end{equation}%
and define the maps $f(z)=z^{M}(z+\alpha )/(z\bar{\alpha}+1)$ and $%
g(z)=f(z^{-1})$. Notice that the above Bethe ansatz equations are invariant
under complex conjugation and $z\rightarrow z^{-1}$, which reflects the $PT$%
-invariance of the Hamiltonian. Thus, we can assume without loss of
generality that there exists a solution $z_{0}$ with $|z_{0}|\leq 1$. For $%
|\alpha |\leq 1$ the image of the closed unit disc under the map $f$ lies
again in the closed unit disc. In contrast the image of the closed unit disc
under the map $g$ lies outside of the open disc. Hence, any solution to (\ref%
{bae2}) with $|\alpha |\leq 1$ must lie on the boundary of these two image
regions, i.e. the unit circle.
\end{proof}\smallskip

\textsc{Remark.} Motivated by the previous proposition we shall henceforth
use the parametrisation%
\begin{equation}
\alpha =\bar{\beta}=ge^{i\theta },\;g\geq 0  \label{boundpar}
\end{equation}%
for the boundary parameters. The Hamiltonian therefore depends now only on
two real parameters, $g\geq 0$ and $0\leq \theta <2\pi $.\smallskip

Note that there are only $M$ relevant solutions to the Bethe ansatz
equations. The trivial roots $z=\pm 1$ do not occur in the spectrum of the
Hamiltonian which allows one to reduce the problem of solving (\ref{bae}) to
finding the roots of the following palindromic or self-reciprocal polynomial%
\begin{equation}
f(z)=z^{2M}f(z^{-1})=z^{2M}+1+(1+g^{2})\sum_{m=1}^{M-1}z^{2m}+2g\cos \theta
\sum_{m=0}^{M-1}z^{2m+1}\ .  \label{f}
\end{equation}%
Thus, all roots $z_{i}$ occur in reciprocal pairs and there exists a unique
polynomial 
\begin{equation}
F(\varepsilon )=\prod_{i=1}^{M}(\varepsilon -\varepsilon _{i})
\end{equation}%
of order $M$ whose roots are given by the single-particle energies (compare
with (\ref{specH3}))%
\begin{equation}
\varepsilon _{i}=z_{i}+z_{i}^{-1}\ .
\end{equation}%
Note that the ambiguity in the definition of the Bethe root $z_{i}$ does not
matter, as the wavefunction (\ref{psi}) is simply rescaled by changing $%
z_{i}\rightarrow z_{i}^{-1}$. Thus, the problem of computing the
eigenvectors and spectrum of the Hamiltonian is reduced to finding the roots 
$\varepsilon _{i}$ of a polynomial $F$ of degree $M$. For example, setting $%
M=8$ and $\rho =\alpha +\beta ,\;\sigma =1+\alpha \beta $ we obtain%
\begin{equation}
F(\varepsilon )=\varepsilon ^{8}-\varepsilon ^{6}(8-\sigma )+5\varepsilon
^{4}(4-\sigma )-2\varepsilon ^{2}(8-3\sigma )+2-\sigma -4\rho \varepsilon
+10\rho \varepsilon ^{3}-6\rho \varepsilon ^{5}+\rho \varepsilon ^{7}\;.
\end{equation}%
We now specialize to the case of particular interest, $\alpha =-\beta
=ig,\;g\in \mathbb{R}$, and present general expressions for the reduced
polynomial $F$ for all $M$.

\subsection{Palindromic polynomials for purely imaginary boundary fields}

For the special choice $\theta =\pm \pi /2$ the palindromic polynomial (\ref%
{f}) simplifies since all terms involving odd powers disappear. One easily
convinces oneself, that this leads to the further simplification that all
roots occur in pairs $\pm \varepsilon _{i}$. Furthermore, if the number of
sites $M$ is odd one easily verifies that one has the roots $z_{i}=\pm \sqrt{%
-1}$. The latter give rise to a \textquotedblleft zero
mode\textquotedblright , $\varepsilon _{i}=0,$ of the Hamiltonian%
\begin{equation*}
\lbrack H,\hat{c}_{z=i}^{\ast }]=0,
\end{equation*}%
where the corresponding wave function is given by%
\begin{equation*}
\psi _{z=i}(x)=\frac{\sin \frac{\pi x}{2}-ig\cos \frac{\pi x}{2}}{\sqrt{%
\frac{M+1}{2}-\frac{M-1}{2}g^{2}}}\;.
\end{equation*}%
A similar expression holds for $z=-i$. As we will see below this solution
for the discrete wave function is connected with a $U(\mathfrak{gl}(1|1))$%
-symmetry of the Hamiltonian. Dividing out the zero mode, we end up with a
palindromic polynomial of degree $2(M-1)$,%
\begin{equation*}
f(z)=\sum_{k=0}^{m}z^{4k}+g^{2}\sum_{k=0}^{m-1}z^{4k+2},\qquad m=\frac{M-1}{2%
}
\end{equation*}%
Setting therefore 
\begin{equation}
m=\left\{ 
\begin{array}{cc}
\frac{M-1}{2}, & M\text{\ odd} \\ 
\frac{M}{2}, & M\text{ even}%
\end{array}%
\right. \;
\end{equation}%
we write once more%
\begin{equation}
f(z)=\prod_{i=1}^{2m}(z^{2}-\varepsilon _{i}z+1),\qquad \varepsilon
_{i}=z_{i}+z_{i}^{-1}
\end{equation}%
for the reduced polynomial. From the two alternative expressions for $f$ we
obtain a linear system of equations for the elementary symmetric polynomials 
$e_{k}=e_{k}(\varepsilon _{1},...,\varepsilon _{m})$ in the roots $%
\varepsilon _{i}$. Solving this system we then define the polynomial%
\begin{equation*}
F(\varepsilon )=\sum_{k=0}^{2m}(-1)^{k}e_{k}\varepsilon
^{m-k}=\prod_{i=1}^{m}(\varepsilon -\varepsilon _{i})(\varepsilon
+\varepsilon _{i})\;.
\end{equation*}%
We now explicitly state this polynomial $F$ for $\alpha =-\beta $ on the
imaginary axis. We have to distinguish the cases of odd and even sites:

\begin{description}
\item[$\boldsymbol{M=2m+1}$.] 
\begin{eqnarray}
F(\varepsilon ) &=&(-)^{m}\sum\limits_{k=0}^{m}\frac{(-)^{k}\varepsilon ^{2k}%
}{(2k+1)!}\left[ \frac{(m+1+k)!}{(m-k)!}-g^{2}\frac{(m+k)!}{(m-1-k)!}\right]
\label{Fodd} \\
&=&2\frac{\sin \left[ (M+1)\arccos \frac{\varepsilon }{2}\right] +g^{2}\sin %
\left[ (M-1)\arccos \frac{\varepsilon }{2}\right] }{\varepsilon \sqrt{%
4-\varepsilon ^{2}}}  \label{invFodd}
\end{eqnarray}

\item[$\boldsymbol{M=2m}$.] 
\begin{eqnarray}
F(\varepsilon ) &=&(-)^{m}\sum\limits_{k=0}^{m}\frac{(-)^{k}\varepsilon ^{2k}%
}{(2k)!}\left[ \frac{(m+k)!}{(m-k)!}-g^{2}\frac{(m-1+k)!}{(m-1-k)!}\right]
\label{Feven} \\
&=&2\frac{\sin \left[ (M+1)\arccos \frac{\varepsilon }{2}\right] +g^{2}\sin %
\left[ (M-1)\arccos \frac{\varepsilon }{2}\right] }{\sqrt{4-\varepsilon ^{2}}%
}  \label{invFeven}
\end{eqnarray}
\end{description}

The expressions in terms of the inverse function $\arccos $ can be checked
by rewriting the Bethe ansatz equations (\ref{bae}). At $\alpha =-\beta =ig$
the latter simplify to the transcendental equation,%
\begin{equation}
1=z^{2M}\frac{\beta +z}{\beta +z^{-1}}\frac{1+\alpha z^{-1}}{1+\alpha z}%
=z^{2M}\frac{z^{2}+g^{2}}{1+z^{2}g^{2}}\;\Leftrightarrow \;g^{2}=-\frac{\sin
[(M+1)\zeta ]}{\sin [(M-1)\zeta ]},\qquad z=\exp (i\zeta )\ .
\label{trigbae}
\end{equation}%
This is in agreement with the above expressions for the single particle
energies as zeroes of (\ref{invFodd}) and (\ref{invFeven}). Note however
that the polynomial expressions (\ref{Fodd}), (\ref{Feven}) are of advantage
in the numerical computation of the single particle energies. From (\ref%
{Fodd}) one also easily spots the occurrence of another zero mode $%
\varepsilon =0$ at $g^{2}=(M+1)/(M-1)$ for $M$ odd which we will connect
below with a representation of the universal enveloping algebra of the Lie
superalgebra $\mathfrak{gl}(1|1)$. We will also exploit the above
polynomials and their expressions in terms of inverse trigonometric
functions to derive approximations for the single particle energies in the
interval $0<g<1$.

\subsection{Groundstate eigenvalues and central charges}

Using the results from the previous section we can expand the single
particle energies as power series in the coupling parameter $g^{2}$.
Employing the equation%
\begin{equation*}
F(\varepsilon )=0
\end{equation*}%
we find the following approximated expressions in the vicinity of the points 
$g^{2}=0$ and $g^{2}=1$ where the exact solutions are known:%
\begin{equation}
0<g\ll 1:\;\varepsilon _{k}=2\cos \frac{\pi k}{M+1}-2g^{2}\frac{\sin \frac{%
\pi k}{M+1}\sin \frac{2\pi k}{M+1}}{M+1}+O(g^{4}),\qquad k=1,...,M
\end{equation}%
and 
\begin{equation}
0\ll g<1:\varepsilon _{k}=2\cos \frac{\pi k}{M}+\left( 1-g^{2}\right) \frac{%
\sin \frac{\pi k}{M}\tan \frac{\pi k}{M}}{M}+O((1-g^{2})^{2}),\qquad
k=1,...,M-1\;.
\end{equation}%
If $M$ is even we have to omit the value $k=M/2$ where the above
approximation is not valid since we then have a zero mode at $g=1$ and the
equation $F(0)=0$ becomes trivial. For $M$ even we thus approximate the two
missing eigenvalues which converge to $\varepsilon =0$ at $g=1$ by%
\begin{equation*}
\varepsilon _{M/2}^{\pm }=\pm 2\sqrt{\frac{2(1-g^{2})}{M(M+2-g^{2}(M-2))}}\;.
\end{equation*}%
This approximation is found by expanding the Bethe ansatz equations (\ref%
{invFeven}) at $\varepsilon =0$. An example, $M=8$, for these approximations
is shown in Figures 2 and 3. The solid lines indicate the exact single
particle energies $\varepsilon _{i}$ and the dashed lines the
approximations. In the vicinity of $g=0$ the approximations are depicted in
Figure 2.%
\begin{equation*}
\includegraphics[scale=0.9]{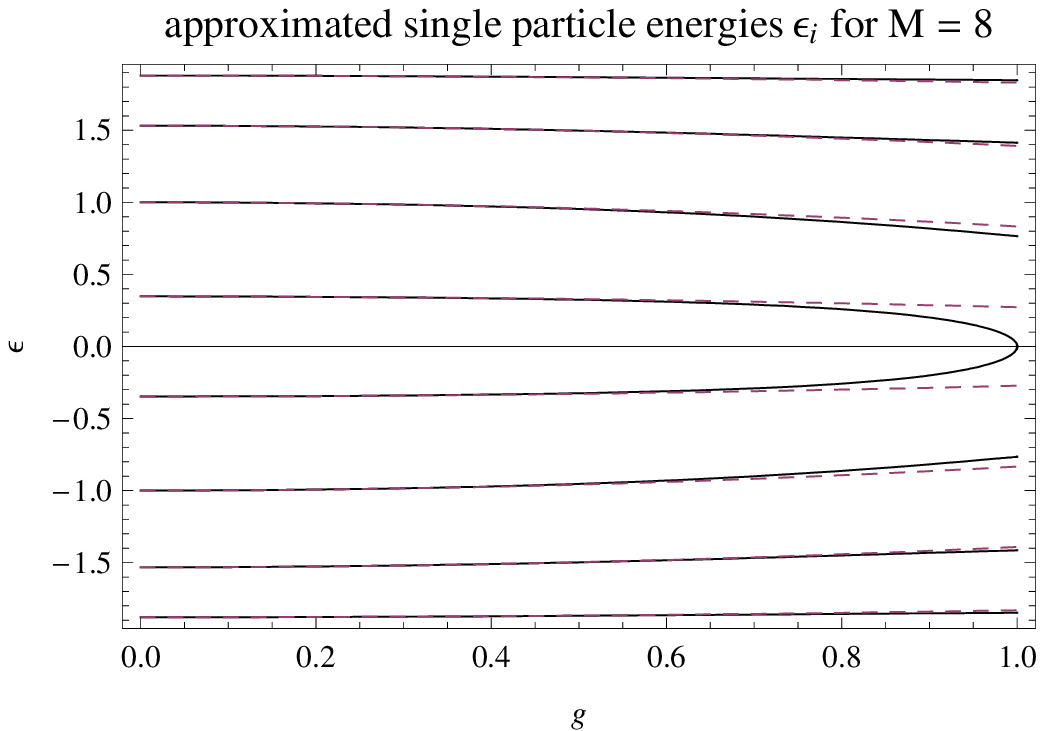}
\end{equation*}%
{\small \medskip }

\begin{center}
{\small Figure 2. Approximations of the single particle energies in the
vicinity of }${\small g}${\small \ = 0.\medskip }
\end{center}

\noindent For the approximations in the vicinity of $g=1$ see Figure 3.%
\begin{equation*}
\includegraphics[scale=0.9]{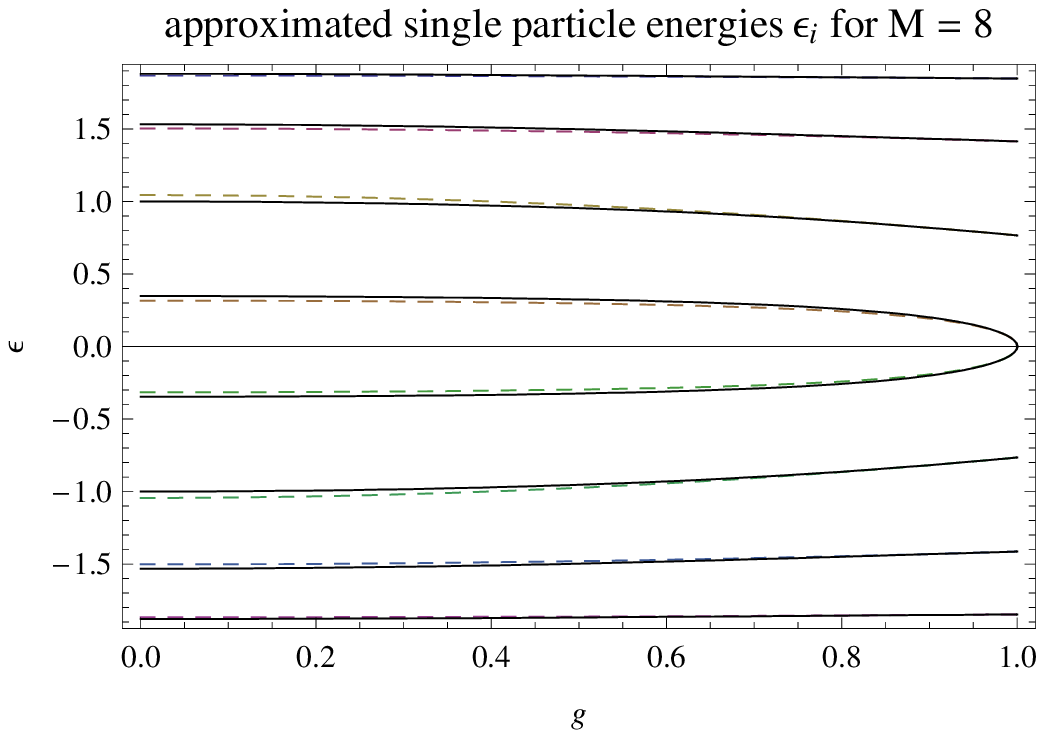}
\end{equation*}%
{\small \medskip }

\begin{center}
{\small Figure 3. Approximations of the single particle energies in the
vicinity of }${\small g}${\small \ = 1.\medskip }
\end{center}

\noindent Having found approximated expressions for the single particle
energies $\varepsilon _{i},$ we can derive approximations for the
groundstate eigenvalue $E_{0}$ of (\ref{Hg}) and discuss the finite size
corrections of the system as it approaches the two points $g=0$ and $g=1$.
For the latter values the result is well known, see Table 1.\medskip

\begin{center}
\begin{tabular}{||l||l|l|}
\hline\hline
& $g=0$ & $g=1$ \\ \hline\hline
$M$ odd & $E_{0}=1-\tfrac{2(M+1)}{\pi }+\tfrac{\pi }{6M}+O(\frac{1}{M^{2}})$
& $E_{0}=1-\tfrac{2M}{\pi }-\tfrac{\pi }{12M}+O(\frac{1}{M^{2}})$ \\ \hline
$M$ even & $E_{0}=1-\tfrac{2(M+1)}{\pi }-\tfrac{\pi }{12M}+O(\frac{1}{M^{2}}%
) $ & $E_{0}=1-\tfrac{2M}{\pi }+\tfrac{\pi }{6M}+O(\frac{1}{M^{2}})$ \\ 
\hline
\end{tabular}%
\medskip

Table 1. Expansion of the groundstate energy with respect to the lattice
size $M$. \medskip
\end{center}

\noindent Using these expansions in the system size $M$ one obtains via the
general formula \cite{BCN86,A86} 
\begin{equation}
E_{0}=2M~f_{\infty }+f_{s}-\frac{\pi c_{eff}}{12M}+O(\frac{1}{M^{2}})
\end{equation}%
the effective central charge $c_{eff}=c-12d_{\min }$ of the conformal field
theory describing the system in the thermodynamic limit. Here $d_{\min }$ is
the smallest scaling dimension occurring in the theory and $f_{\infty
},~f_{s}$ are the bulk and surface free energy, respectively. One recovers
the familiar central charges at $g=1$ mentioned in the introduction: for an
odd number of sites the conformal anomaly is $c=1$ ($d_{\min }=0$) while it
is $c=-2$ for $M$ even ($d_{\min }=0$). Note that for $M$ odd at $g=0$ we
also formally obtain $c_{eff}=-2$, it appears however that this case has not
been investigated further in the literature \cite{F93}.

Let us now consider the approximate expressions for $0<g<1$. Summing the
above expressions for the single particle energies we arrive at

\begin{description}
\item[$M\;$odd, $0<g\ll 1:$] 
\begin{align}
E_{0}& =-\sum_{k=1}^{\frac{M-1}{2}}\varepsilon _{k}=1-\dfrac{\cos \frac{\pi 
}{2(M+1)}}{\sin \frac{\pi }{2(M+1)}}+\frac{2\cos \frac{\pi }{M+1}\cos \frac{%
\pi }{2(M+1)}}{(M+1)\sin \frac{3\pi }{2(M+1)}}~g^{2}+O(g^{4})  \notag \\
& =1+\frac{4g^{2}}{3\pi }-\dfrac{2(M+1)}{\pi }+\dfrac{\pi }{6M}%
+O(M^{-2})+O(g^{4})
\end{align}

\item[$M$ even, $0<g\ll 1:$] 
\begin{align}
E_{0}& =-\sum_{k=1}^{\frac{M}{2}}\varepsilon _{k}=1-\dfrac{1}{\sin \frac{\pi 
}{2(M+1)}}+\frac{1+\cos \frac{\pi }{M+1}}{(M+1)\sin \frac{3\pi }{2(M+1)}}%
~g^{2}+O(g^{4})  \notag \\
& =1+\frac{4g^{2}}{3\pi }-\dfrac{2(M+1)}{\pi }-\dfrac{\pi }{12M}%
+O(M^{-2})+O(g^{4})
\end{align}
\end{description}

Thus in both cases we only see a change in the surface term $f_{s}$ up to
order $g^{2}$ in the vicinity of $g=0$. A computation of the second
derivative $\varepsilon _{k}^{\prime \prime }(g^{2}=0)$ using the equation $%
F(\varepsilon )=0$ shows that also the $g^{4}$ contribution to the
groundstate only contributes to the surface energy. The computation of
higher orders becomes cumbersome and it might be preferrable to rely on
field theoretic methods instead, we briefly comment on this in the
conclusions.

Let us now turn to the case when $g$ is the vicinity of $g=1$ to see whether
we encounter a shift in the central charge here. Again we sum the
approximate expressions of the single particle energies to obtain the
expression for the groundstate eigenvalue. Unlike for $g=0$ we cannot find
exact expressions at this point. Instead we split the sum over the
correction terms for $g<1$ into two parts,%
\begin{equation*}
\sum_{k}\frac{\sin \frac{\pi k}{M}\tan \frac{\pi k}{M}}{M}=\frac{1}{M}%
\sum_{k}\frac{1}{\cos \frac{\pi k}{M}}-\frac{1}{M}\sum_{k}\cos \frac{\pi k}{M%
}\ ,
\end{equation*}%
and for $M>4$ employ the Euler-Maclaurin formula to obtain the following
approximations for the first sum,%
\begin{equation}
M\;\text{odd, }0\ll g<1:\frac{1}{M}\sum_{k=1}^{\frac{M-1}{2}}\frac{1}{\cos 
\frac{\pi k}{M}}\approx \frac{4-3\ln \frac{1}{M}+3\ln \frac{4}{\pi }}{3\pi }-%
\frac{1}{2M}+O(\frac{1}{M^{2}})
\end{equation}%
and%
\begin{equation}
M\;\text{even, }0\ll g<1:\frac{1}{M}\sum_{k=1}^{\frac{M}{2}-1}\frac{1}{\cos 
\frac{\pi k}{M}}\approx \frac{7-12\ln \frac{1}{M}+12\ln \frac{2}{\pi }}{%
12\pi }-\frac{1}{2M}+O(\frac{1}{M^{2}})\;.
\end{equation}%
For $M$ even we have another contribution from the energy level which
becomes a zero mode at $g=1$. Keeping $\delta =1-g^{2}$ fixed and choosing $%
M $ large enough such that $\delta M\gg 1$ we find the asymptotic expansion%
\begin{equation}
\varepsilon _{M/2}=\frac{2\sqrt{2}}{M}-\frac{2\sqrt{2}(2-\delta )}{\delta
M^{2}}+~O(\frac{1}{\delta ^{2}M^{3}}),\qquad \delta :=1-g^{2}
\end{equation}%
However, for our purposes this is not a good approximation since it should
vanish at $g=1$. If we therefore expand first with respect to $\delta $ we
find instead, 
\begin{equation}
\varepsilon _{M/2}=\sqrt{\frac{2\delta }{M}}-\frac{(M-2)}{4\sqrt{2M}}~\delta
^{3/2}+O(\delta ^{5/2})\ .
\end{equation}%
From this approximation we can not infer the correct finite size scaling
properties. In order to extract conclusive results about the finite size
scaling behaviour we therefore need to find the exact solutions for the
single particle energies first. For now we have to leave this problem open,
but we hope to address it in future work by field theoretic methods and
exploiting the underlying algebraic structures which we highlight next.

\section{Quasi-Fermions and $U(\mathfrak{gl}(1|1))$-invariance}

As we saw in the previous section the spectrum of the Hamiltonian is
composed out of single particle excitations with energies $\varepsilon
_{i}=z_{i}+z_{i}^{-1}$ which are \textquotedblleft
created\textquotedblright\ from a pseudo-vacuum applying the creation
operator (\ref{ccross}). Naturally, we wish to find the corresponding
annihilation operator in order to set up a Fermion or CAR algebra in
quasi-momentum space which diagonalises the Hamiltonian. As the Hamiltonian
is non-Hermitian the annihilation operator of momentum $k=-i\ln z$ cannot
simply be the Hermitian adjoint of (\ref{ccross}) with respect to the
original inner product. Instead we introduce the time-reversed annihilation
operator 
\begin{equation}
\hat{d}_{z}:=T\hat{c}_{z}T=\sum_{x=1}^{M}\psi _{z}(x;\alpha ,\beta =\bar{%
\alpha})c_{x},  \label{d}
\end{equation}%
where the conjugation with the time reversal operator is motivated by
observing that 
\begin{equation*}
THT=H^{\ast }\qquad \text{for}\qquad \beta =\bar{\alpha}~.
\end{equation*}%
Obviously, we then obtain the desired commutation relation%
\begin{equation}
\lbrack H,\hat{d}_{z}]=(z+z^{-1})\hat{d}_{z}\ ,  \label{specH2}
\end{equation}%
i.e. $\hat{d}_{z}$ annihilates a quasi-particle of energy $\varepsilon
=z+z^{-1}$.

Alternatively, we could also have employed the parity operator $P$ in the
definition (\ref{d}), since\ $PHP=H^{\ast }$. However, up to a phase factor
this would lead to the same result due to the following identities of the
wave function $\psi _{z}(x;\alpha ):=\psi _{z}(x;\alpha ,\beta =\bar{\alpha}%
) $,%
\begin{equation}
\psi _{z}(M+1-x;\alpha )=-z^{M+1}~\frac{1+\bar{\alpha}/z}{1+\bar{\alpha}z}%
~\psi _{z}(x;\bar{\alpha})=-z^{M+1}\overline{\psi _{z}(x;\alpha )}\ .
\label{PTpsi}
\end{equation}%
Here we have assumed that $z$ is a solution of (\ref{bae}) and hence, lies
on the unit circle. From (\ref{PTpsi}) together with (\ref{PTRc}) we infer
that%
\begin{equation}
PT\hat{c}_{z}^{\ast }PT=-z^{-M-1}\hat{c}_{z}^{\ast }\qquad \text{and\qquad }P%
\hat{c}_{z}P=-z^{-M-1}\hat{d}_{z}\ .  \label{Pcd}
\end{equation}%
This in particular implies that the system possesses $PT$-symmetry, i.e. the
eigenstates of the Hamiltonian can be chosen to be eigenvectors of the $PT$%
-operator. However, this fact does not guarantee quasi-Hermiticity which we
show next.

Our strategy for establishing quasi-Hermiticity is to turn the operators (%
\ref{ccross}) and (\ref{d}) into a well defined representation of the
fermionic oscillator or CAR algebra. Under an appropriate renormalization of
the wavefunction (\ref{psi}),%
\begin{equation}
\psi _{z}(x;\alpha )\rightarrow \frac{a_{z}\psi _{z}(x;\alpha )}{%
(-2Ma_{z}a_{z^{-1}}+[M]_{z}(z^{M+1}a_{z^{-1}}^{2}+z^{-M-1}a_{z}^{2}))^{\frac{%
1}{2}}},\qquad a_{z}:=1+\alpha /z,
\end{equation}%
one verifies that the following anticommutation relations hold, 
\begin{equation}
\lbrack \hat{c}_{z_{1}}^{\ast },\hat{c}_{z_{2}}^{\ast }]_{+}=[\hat{d}%
_{z_{1}},\hat{d}_{z_{2}}]_{+}=0\qquad \text{and}\qquad \lbrack \hat{c}%
_{z_{1}}^{\ast },\hat{d}_{z_{2}}]_{+}=\delta _{z_{1},z_{2}}\ .  \label{car}
\end{equation}%
Here $z_{1},z_{2}$ are two solutions to the Bethe ansatz equations (\ref{bae}%
). The anticommutation relations (\ref{car}) are a direct consequence of the
wave function identities%
\begin{equation}
\sum_{x}\psi _{z_{k}}(x)\psi _{z_{l}}(x)=\delta _{kl}\quad \quad \text{and}%
\quad \quad \sum_{k}\psi _{z_{k}}(x)\psi _{z_{k}}(y)=\delta _{x,y},
\label{comp}
\end{equation}%
where the index $k$ labels the solutions to the Bethe ansatz equations.
Notice that these identities lead to the inversion formulae%
\begin{equation}
c_{x}^{\ast }=\sum_{k}\psi _{z_{k}}(x)\hat{c}_{z_{k}}^{\ast }=\sum_{k}%
\overline{\psi _{z_{k}}(x)}\hat{d}_{z_{k}}^{\ast }
\end{equation}%
and%
\begin{equation}
c_{x}=\sum_{k}\overline{\psi _{z_{k}}(x)}\hat{c}_{z_{k}}=\sum_{k}\psi
_{z_{k}}(x)\hat{d}_{z_{k}}\ .
\end{equation}%
While we have obtained the correct anticommuation relations in
quasi-momentum space, the algebras generated by $\{\hat{c}_{z}^{\ast
},d_{z}\}$ do not possess the right $\ast $-structure (anti-involution),
where creation and annihilation operators are related by Hermitian
conjugation. However, we now introduce the quasi-Hermiticity operator and
the associated inner product with respect to which (\ref{ccross}) and (\ref%
{d}) possess the right $\ast $-structure.

\begin{theorem}
The Hamiltonians (\ref{H}), (\ref{Hprime}) are quasi-Hermitian for $\alpha =%
\bar{\beta}$ and $|\alpha |<1$. Let $z_{j}=\exp ik_{j},\;j=1,...,M$ be the $%
M $ roots of the polynomial (\ref{f}) not containing a reciprocal pair ($%
k_{j}\neq -k_{j^{\prime }}$ for all $j,j^{\prime }$), then%
\begin{equation}
H=-\frac{\alpha +\bar{\alpha}}{2}-\sum_{j=1}^{M}2\cos k_{j}~\hat{c}%
_{k_{j}}^{\ast }\hat{d}_{k_{j}}\ .
\end{equation}%
A similar expression holds for (\ref{Hprime}). The quasi-Hermiticity
operator $\eta :V^{\otimes M}\rightarrow V^{\otimes M}$ and its inverse are
given by%
\begin{equation}
\eta =\sum_{n=0}^{M}\tsum_{k_{j_{1}}<...<k_{j_{n}}}\hat{d}_{k_{j_{1}}}^{\ast
}\cdots \hat{d}_{k_{j_{n}}}^{\ast }|0\rangle \langle 0|\hat{d}%
_{k_{j_{n}}}\cdots \hat{d}_{k_{j_{1}}}  \label{fermieta}
\end{equation}%
$\quad $and\quad 
\begin{equation}
\eta ^{-1}=\sum_{n=0}^{M}\tsum_{k_{j_{1}}<...<k_{j_{n}}}\hat{c}%
_{k_{j_{1}}}^{\ast }\cdots \hat{c}_{k_{j_{n}}}^{\ast }|0\rangle \langle 0|%
\hat{c}_{k_{j_{n}}}\cdots \hat{c}_{k_{j_{1}}}\;.
\end{equation}%
Besides quasi-Hermiticity of the Hamiltonian the map $\eta $ also induces
the conventional $\ast $-structure of free fermions in quasi-momentum space
via the identities%
\begin{equation}
\eta \hat{c}_{k}^{\ast }=\hat{d}_{k}^{\ast }\eta \qquad \text{and}\qquad
\eta \hat{d}_{k}=\hat{c}_{k}\eta ,\qquad k=1,...,M\ .
\end{equation}
\end{theorem}

\begin{proof}
First we notice that the creation operators (\ref{ccross}) evaluated at the
solutions $z=z_{i}$ of the equations (\ref{bae}) provide us with a basis in
the state space, see (\ref{particlebasis}). Due to the symmetry of the equations (\ref{bae}) under the
replacement $z\rightarrow z^{-1}$, there are only $M$ relevant distinct
solutions, despite the fact that the polynomial order is $2M+2.$ According
to the anticommutation relations (\ref{car}) these solutions then yield $%
2^{M}=\dim V^{\otimes M}$ eigenvectors. The assertions then follow from the
previous proposition showing that all Bethe roots lie on the unit circle for
$\alpha =\bar{\beta}$ and $|\alpha |<1$ as well as employing (\ref{specH3})
and (\ref{car}).
\end{proof}

\begin{corollary}
The C-operator $C=P\eta $ has the following expression in terms of creation
and annihilation operators,%
\begin{equation}
C=\sum_{n=0}^{M}\tsum_{k_{j_{1}}<...<k_{j_{n}}}(-)^{n}e^{i(M+1)(k_{j_{1}}+%
\cdots +k_{j_{n}})}\hat{c}_{k_{j_{1}}}^{\ast }\cdots \hat{c}%
_{k_{j_{n}}}^{\ast }|0\rangle \langle 0|\hat{d}_{k_{j_{n}}}\cdots \hat{d}%
_{k_{j_{1}}}\ .
\end{equation}%
Thus, with respect to the quasi-particle basis $\{\left\vert
k_{j_{1}},...,k_{j_{n}}\right\rangle =\hat{c}_{k_{j_{1}}}^{\ast }\cdots \hat{%
c}_{k_{j_{n}}}^{\ast }\left\vert 0\right\rangle \}_{n=0}^{M}$ the $C$%
-operator is simply given in terms of the total quasi-momentum operator $%
\hat{P}=\sum_{r}k_{j_{r}}$ and the quasi-particle number $\hat{N}=\sum_{j}%
\hat{c}_{k_{j}}^{\ast }\hat{d}_{k_{j}}$ as%
\begin{equation}
C=(-1)^{\hat{N}}e^{i(M+1)\hat{P}}\;.
\end{equation}
\end{corollary}

\begin{proof}
An immediate consequence from (\ref{Pcd}).
\end{proof}\medskip

Let us summarize the result: the fermion or CAR algebra $\{c_{x},c_{x}^{\ast
}\}_{x=1}^{M}$ in position space with respect to the original inner product
is replaced by a CAR algebra $\{\hat{d}_{k},\hat{c}_{k}^{\ast }\}_{k=1}^{M}$
in quasi-momentum space with respect to the $\eta $-product,%
\begin{equation}
\langle \hat{c}_{k}^{\ast }v,w\rangle _{\eta }=\langle v,\hat{c}_{k}\eta
w\rangle =\langle v,\hat{d}_{k}w\rangle _{\eta }
\end{equation}%
and%
\begin{equation}
\langle \hat{d}_{k}v,w\rangle _{\eta }=\langle v,\hat{d}_{k}^{\ast }\eta
w\rangle =\langle v,\hat{c}_{k}^{\ast }w\rangle _{\eta },\qquad v,w\in
V^{\otimes M}\ .
\end{equation}%
Note that there is another copy of this CAR algebra, namely $\{\hat{d}%
_{k}^{\ast },\hat{c}_{k}\}_{k=1}^{M}$ but with respect to the $\eta ^{-1}$%
-product.

We have excluded the unit circle from the allowed range of the boundary
parameter $\alpha $. While the spectrum of $H$ is also real in this case, it
does not necessarily imply the existence of a positive \emph{definite}
operator $\eta $, i.e. quasi-Hermiticity is a stronger condition. This is
linked to the appearance of Jordan blocks in the Jordan normal form of the
Hamiltonian as we will discuss now for the special case when $\alpha =-\beta
=ig,\;g\geq 0$. Without loss of generality we can restrict ourselves to
positive values of $g$, since $g$ and $-g$ are related via spin-reversal. We
will show that the appearance of non-trivial Jordan blocks coincides with
the vanishing of a central element in a representation of a specific
superalgebra.

\subsection{\textquotedblleft Deformed\textquotedblright quantum group
symmetry on the imaginary axis}

Since the Hamiltonians (\ref{H}), (\ref{Hprime}) and (\ref{Hg}) can be
expressed in terms of free fermions as we have just seen, there are several
symmetries and algebraic structures associated with them. These have been
previously investigated for $\alpha ,\beta $ on the unit circle. Here we
shall show that these symmetries and algebras can be extended inside the
unit disc along the imaginary axis setting $\alpha =-\beta =ig,\;g\geq 0$.
We shall also relate the new inner product to the representations of these
algebras. We start by reviewing the known symmetries and algebras for the
case $g=1$, i.e. $\alpha ,\beta $ on the unit circle.

It is well-known that for $g=1$ the Hamiltonian (\ref{Hg}) possesses a $%
U_{q}(sl_{2})$-symmetry with $q=i\equiv \sqrt{-1}$ \cite{PS90}. Below we
shall see that for an odd number of sites, $M\in 2\mathbb{N}+1$, this
symmetry can be extended along the imaginary axis for $0\leq g<\sqrt{%
(M+1)/(M-1)}$. We also comment on a deformation of the Temperley-Lieb
algebra \cite{TL71,Jo83} (see \cite{PPM91} for a text book), which is the
commutant of $U_{q}(sl_{2})$. First we recall the basic algebraic
definitions.

\begin{definition}
The q-deformed enveloping algebra $U_{q}(sl_{2})$ is defined in terms of the
Chevalley generators $\{E,F,K^{\pm 1}\}$ and the relations%
\begin{equation}
KK^{-1}=K^{-1}K=1,\qquad KE=q^{2}EK,\qquad KF=q^{-2}FK,\qquad \lbrack E,F]=%
\frac{K-K^{-1}}{q-q^{-1}}\;.  \label{Usl2def}
\end{equation}
\end{definition}

\begin{definition}
The Temperley-Lieb algebra $TL_{M}(q)$ is the associative, unital algebra
generated by $\{e_{i}\}_{i=1}^{M-1}$ subject to the identities%
\begin{equation}
e_{i}^{2}=-(q+q^{-1})e_{i},\qquad e_{i}e_{i\pm 1}e_{i}=e_{i},\qquad
e_{i}e_{j}=e_{j}e_{i},\quad |i-j|>1\;.  \label{TLdef}
\end{equation}
\end{definition}

Setting the deformation parameter to $q=\sqrt{-1}$ we obtain at $g=1$ the
following representations of the two algebras%
\begin{equation}
K^{\pm 1}\mapsto i^{\pm M}\prod_{x=1}^{M}\sigma _{x}^{z},\qquad E\mapsto
\sum_{x}i^{x-1}c_{x}^{\ast },\qquad F\mapsto \sum_{x}i^{x}c_{x}K^{-1}
\label{unitrep}
\end{equation}%
and%
\begin{equation}
e_{x}\mapsto c_{x}c_{x+1}^{\ast }-c_{x}^{\ast
}c_{x+1}+i(n_{x}-n_{x+1}),\qquad x=1,2,...,M-1\;.  \label{unitTLrep}
\end{equation}%
One easily verifies that the above algebraic relation of $U_{q}(sl_{2})$ and 
$TL_{M}(q)$ are satisfied. In addition we have the identities%
\begin{equation}
E^{2}=F^{2}=0\qquad \text{and}\qquad K^{2}=(-1)^{M}\;.  \label{centre}
\end{equation}%
Moreover, we have that the action of the Temperley-Lieb and the quantum
algebra commute (this is a special case of the quantum analogue of
Schur-Weyl duality \cite{Jim86b}),%
\begin{equation}
\lbrack U_{q}(sl_{2}),TL_{M}(q)]=0,
\end{equation}%
whence the Hamiltonian%
\begin{equation}
g=1:\qquad H=H^{\prime }=\sum_{x=1}^{M-1}e_{x}
\end{equation}%
is quantum group invariant. For $g=1$ these relations hold for both, odd and
even numbers of sites. However, note that we have%
\begin{equation}
\lbrack E,F]=\left\{ 
\begin{array}{cc}
0, & M\text{ even} \\ 
i^{M-1}K, & M\text{ odd}%
\end{array}%
\right. \;.
\end{equation}%
As we shall see below the vanishing of the commutator for $M$ even is
closely related to the fact that the discrete wave function of an associated
fermionic zero mode has vanishing norm, 
\begin{equation*}
\sum_{x}\psi _{k=i}(x)^{2}=0\;.
\end{equation*}%
This has profound consequences for the quasi-Hermiticity of the Hamiltonian,
since the quasi-Hermiticity operator $\eta $ ceases to be positive \emph{%
definite}, $\eta >0,$ and instead becomes positive \emph{semi-definite}, $%
\eta \geq 0$. One of the novel results in this article is the observation
that a similar scenario also happens for $M$ odd, albeit at the value $g=%
\sqrt{\frac{M+1}{M-1}}>1$. In order to discuss these issues and to interpret
the new inner product in terms of representation theory it is favourable to
introduce yet another algebra.

\begin{definition}
The universal enveloping algebra $U(\mathfrak{gl}(1|1))$ of the superalgebra 
$\mathfrak{gl}(1|1)$ is the $\mathbb{Z}_{2}$-graded associative algebra
(over $\mathbb{C}$) generated by the odd (fermionic) elements $\{X^{\pm }\}$
and the even (bosonic) elements $\{Y,Z\}$ satisfying the relations%
\begin{equation}
\lbrack Y,X^{\pm }]=\pm X^{\pm },\qquad \lbrack Z,Y]=[Z,X^{\pm }]=0,\qquad
\lbrack X^{+},X^{-}]_{+}=Z\;.  \label{gl11def}
\end{equation}%
There is a natural anti-involution or $\ast $-structure on $U(\mathfrak{gl}%
(1|1))$ by setting%
\begin{equation}
(X^{\pm })^{\ast }=X^{\mp },\qquad Y^{\ast }=Y,\qquad Z^{\ast }=Z\;.
\label{*gl11}
\end{equation}
\end{definition}

The following $U(\mathfrak{gl}(1|1))$-representation at $g=1$ is a special
case of the representation discussed in \cite{HR92} for the deformed case $%
U_{q}(\mathfrak{gl}(1|1))$ on which we will comment later,%
\begin{equation}
X^{+}\mapsto \tsum_{x}i^{x-1}c_{x}^{\ast },\quad X^{-}\mapsto
\tsum_{x}i^{x-1}c_{x},\quad Y\mapsto S^{z},\quad Z\mapsto \left\{ 
\begin{array}{cc}
0, & M\text{ even} \\ 
M~\mathbf{1}, & M\text{ odd}%
\end{array}%
\right. \;.  \label{g1rep}
\end{equation}%
Again, we observe a crucial difference for the two cases $M$ even and odd:
the vanishing of the central element $Z$. We will now discuss for $g\neq 1$
an extension of the above representation. For $M$ even we shall see that (%
\ref{g1rep}) is a singular limit of a representation with $Z\neq 0$.
However, we start our discussion with the case $M$ odd, since it is more
well-behaved and shall then move on to $M$ even pointing out similarities
and differences for $0\leq g\leq 1$.

\subsection{Odd number of sites}

For $M$ odd the quantum group symmetry can be extended. As noted earlier we
have for $\theta =\pi /2$ and $M$ odd the following two Bethe roots $%
z_{1}=\pm \sqrt{-1}$ giving rise to a zero mode of the Hamiltonian (we shall
focus on $z=i$)%
\begin{equation}
\lbrack H,\hat{c}_{z=i}^{\ast }]=[H,\hat{d}_{z=i}]=0,  \label{0mode}
\end{equation}%
where the corresponding wave function is given by%
\begin{equation}
\psi _{z=i}(x)=\frac{\sin \frac{\pi x}{2}-ig\cos \frac{\pi x}{2}}{\sqrt{%
\frac{M+1}{2}-\frac{M-1}{2}g^{2}}}\;.  \label{0wave}
\end{equation}%
This solution prompts the following operator definitions%
\begin{equation}
U=\sum_{x}\sin \frac{\pi x}{2}~c_{x}=\sum_{x\text{ odd}}(-1)^{\frac{x-1}{2}%
}c_{x}  \label{U}
\end{equation}%
and%
\begin{equation}
V=\sum_{x}\cos \frac{\pi x}{2}~c_{x}=\sum_{x\text{ even}}(-1)^{\frac{x}{2}%
}c_{x}  \label{V}
\end{equation}%
such that%
\begin{equation}
\hat{c}_{z=i}^{\ast }=\frac{U^{\ast }-igV^{\ast }}{\sqrt{\frac{M+1}{2}-\frac{%
M-1}{2}g^{2}}}\qquad \text{and}\qquad \hat{d}_{z=i}=\frac{U-igV}{\sqrt{\frac{%
M+1}{2}-\frac{M-1}{2}g^{2}}}\;.
\end{equation}%
The operators $U,V$ satisfy the relations%
\begin{equation}
U^{2}=V^{2}=[U,V]_{+}=[U,V^{\ast }]_{+}=[U^{\ast },V]_{+}=0  \label{UVrel1}
\end{equation}%
and%
\begin{equation}
\lbrack U,U^{\ast }]_{+}=\frac{M+1}{2}~\mathbf{1},\qquad \lbrack V,V^{\ast
}]_{+}=\frac{M-1}{2}~\mathbf{1}\;.  \label{UVrel2}
\end{equation}%
Thus, up to a trivial renormalization, $U\rightarrow U/\sqrt{(M+1)/2}$ and $%
V\rightarrow V/\sqrt{(M-1)/2}$, we can think of $U,V$ as two fermionic
oscillators. Alternatively, we can also interpret them as two
representations of the \emph{non-deformed} superalgebra $\mathfrak{gl}(1|1)$
introduced above by identifying%
\begin{equation}
Y\mapsto S^{z},\qquad Z\mapsto \frac{M+1}{2}~\mathbf{1},\qquad X^{+}\mapsto
U^{\ast },\qquad X^{-}\mapsto U  \label{Urep}
\end{equation}%
and%
\begin{equation}
Y\mapsto S^{z},\qquad Z\mapsto \frac{M-1}{2}~\mathbf{1},\qquad X^{+}\mapsto
V^{\ast },\qquad X^{-}\mapsto V\;.  \label{Vrep}
\end{equation}%
We thus obtain two distinct representations of $\mathfrak{gl}(1|1)$ which in
addition preserve the $\ast $-structure (\ref{*gl11}) with respect to the
original inner product. For $0<g<1$ neither of these two representations by
itself give rise to a symmetry of the Hamiltonian, instead we have to
consider the combined representation%
\begin{equation}
Y\mapsto S^{z},\qquad Z\mapsto \left( \tfrac{M+1}{2}-g^{2}\tfrac{M-1}{2}%
\right) ~\mathbf{1},\qquad X^{+}\mapsto U^{\ast }-igV^{\ast },\qquad
X^{-}\mapsto U-igV\;.  \label{gl11rep}
\end{equation}%
According to (\ref{0mode}) the generators $X^{\pm }$ create respectively
annihilate zero modes of the Hamiltonian and we therefore have an $U(%
\mathfrak{gl}(1|1))$-symmetry of the Hamiltonian (\ref{Hg}),%
\begin{equation}
\lbrack H_{g},U(\mathfrak{gl}(1|1))]=0,\qquad 0\leq g<\sqrt{\frac{M+1}{M-1}}%
\;.
\end{equation}%
Note that in the representation (\ref{gl11rep}) the $\ast $-structure (\ref%
{*gl11}) is not preserved with respect to the original inner product,%
\begin{equation}
X^{\pm }\neq (X^{\mp })^{\ast },
\end{equation}%
but with respect to the new inner product induced by $\eta $,%
\begin{equation}
\eta X^{\pm }=(X^{\mp })^{\ast }\eta \;.
\end{equation}%
This allows one to interpret the new Hilbert space structure induced by $%
\eta $ in a purely representation theoretic setting. It also singles out the 
$U(\mathfrak{gl}(1|1))$-symmetry over the $U_{q}(sl_{2})$-symmetry whose
extension we discuss next.

Namely, for $M\in 2\mathbb{N}+1$, $0\leq g<\sqrt{\frac{M+1}{M-1}},$ and $q=%
\sqrt{-1}$ we now set%
\begin{equation}
K^{\pm 1}\mapsto i^{\pm M}\prod_{x=1}^{M}\sigma _{x}^{z},\qquad E\mapsto 
\frac{U^{\ast }-igV^{\ast }}{\sqrt{\frac{M+1}{2}-\frac{M-1}{2}g^{2}}},\qquad
F\mapsto \frac{-(gV+iU)K}{\sqrt{\frac{M+1}{2}-\frac{M-1}{2}g^{2}}}\ .
\label{QGodd}
\end{equation}%
Employing the anticommutation relations for $U,V$ and their Hermitian
adjoints, one easily verifies that this representation is well-defined and
that the relations (\ref{centre}) continue to hold. Apparently, we recover
the familiar representation (\ref{unitrep}) on the unit circle in the limit $%
g\rightarrow 1$.

The obvious guess for an extension of the Temperley-Lieb algebra for $g\geq
0 $,%
\begin{equation}
e_{x}\mapsto c_{x}c_{x+1}^{\ast }-c_{x}^{\ast
}c_{x+1}+ig(n_{x}-n_{x+1}),\qquad x=1,2,...,M-1\;,  \label{TLg}
\end{equation}%
does yield the modified commutation relations%
\begin{equation}
e_{x}^{2}=(1-g^{2})[(1-n_{x})n_{x+1}+n_{x}(1-n_{x+1})]
\end{equation}%
and%
\begin{equation}
e_{x}e_{x\pm 1}e_{x}=g^{2}e_{x}+ig(1-g^{2})(n_{x}-n_{x+1})[1+(n_{x\pm
1}-n_{x+1\pm 1})(n_{x}-n_{x+1})]\ .
\end{equation}%
This extension of the Temperley-Lieb algebra in terms of \textquotedblleft
local Hamiltonians\textquotedblright , i.e. the nearest neighbour terms,
does not preserve the algebraic structure at $g=1$. Furthermore, Schur-Weyl
duality is broken for $g<1$: the commutation relations between the fermionic
oscillators and and the extended Temperley-Lieb generators (\ref{TLg}) are%
\begin{eqnarray}
\lbrack U,e_{x}] &=&\left\{ 
\begin{array}{cc}
(-)^{\frac{x+1}{2}}(c_{x+1}-igc_{x}), & x\text{ odd} \\ 
(-)^{\frac{x+2}{2}}(c_{x}+igc_{x+1}), & x\text{ even}%
\end{array}%
\right. , \\
\lbrack V,e_{x}] &=&\left\{ 
\begin{array}{cc}
(-)^{\frac{x-1}{2}}(c_{x}+igc_{x+1}), & x\text{ odd} \\ 
(-)^{\frac{x+2}{2}}(c_{x+1}-igc_{x}), & x\text{ even}%
\end{array}%
\right. ,
\end{eqnarray}%
whence we now have for the quantum group generators the identities%
\begin{equation}
0\leq g\leq \sqrt{\frac{M+1}{M-1}}:\;[E,e_{x}+e_{x+1}]=[F,e_{x}+e_{x+1}]=0,%
\qquad x=1,3,5,...,M-2\;.
\end{equation}%
Thus, while for $M\in 2\mathbb{N}+1$ the Hamiltonian remains quantum group
invariant for all values $0\leq g\leq \sqrt{\frac{M+1}{M-1}}$, we now have
to consider pairs of the extended Temperley-Lieb generators.

Note that the representation (\ref{QGodd}) becomes singular at $%
g^{2}=(M+1)/(M-1)$, this is also the value where the fermionic modes $%
X^{+}=U^{\ast }-igV^{\ast }$ and $X^{-}=U-igV$ anticommute$,$%
\begin{equation}
g^{2}=\frac{M+1}{M-1}:\qquad \lbrack X^{+},X^{-}]_{+}=Z=0\;.
\end{equation}%
This is precisely the scenario mentioned above for $M$ even at $g=1$. We see
that the associated wavefunction (\ref{0wave}) becomes singular as its norm
vanishes. Moreover, the Jordan normal form $J$ of the Hamiltonian possesses
now non-trivial $3\times 3$ blocks. For instance, we have for $M=5$ and the
sector $S^{z}=1/2$ the Jordan normal form,%
\begin{equation*}
J=\left( 
\begin{array}{cccccccccc}
0 & 0 & 0 & 0 & 0 & \cdots &  &  & \cdots & 0 \\ 
0 & 0 & 1 & 0 & 0 &  &  &  &  & \vdots \\ 
\vdots & 0 & 0 & 1 & 0 &  &  &  &  &  \\ 
& \vdots & 0 & 0 & 0 & 0 &  &  &  &  \\ 
&  & \vdots & 0 & -\sqrt{\frac{5}{2}} & 1 & 0 &  &  &  \\ 
&  &  & \vdots & 0 & -\sqrt{\frac{5}{2}} & 1 & 0 &  & \vdots \\ 
&  &  &  & \vdots & 0 & -\sqrt{\frac{5}{2}} & 0 & 0 & 0 \\ 
&  &  &  &  & \vdots & 0 & \sqrt{\frac{5}{2}} & 1 & 0 \\ 
\vdots &  &  &  &  &  & 0 & 0 & \sqrt{\frac{5}{2}} & 1 \\ 
0 & \cdots &  &  &  &  & 0 & 0 & 0 & \sqrt{\frac{5}{2}}%
\end{array}%
\right) ~.
\end{equation*}%
Some of the eigenvalues smoothly join up as can be seen in Figure 4.%
\begin{equation*}
\includegraphics[scale=0.9]{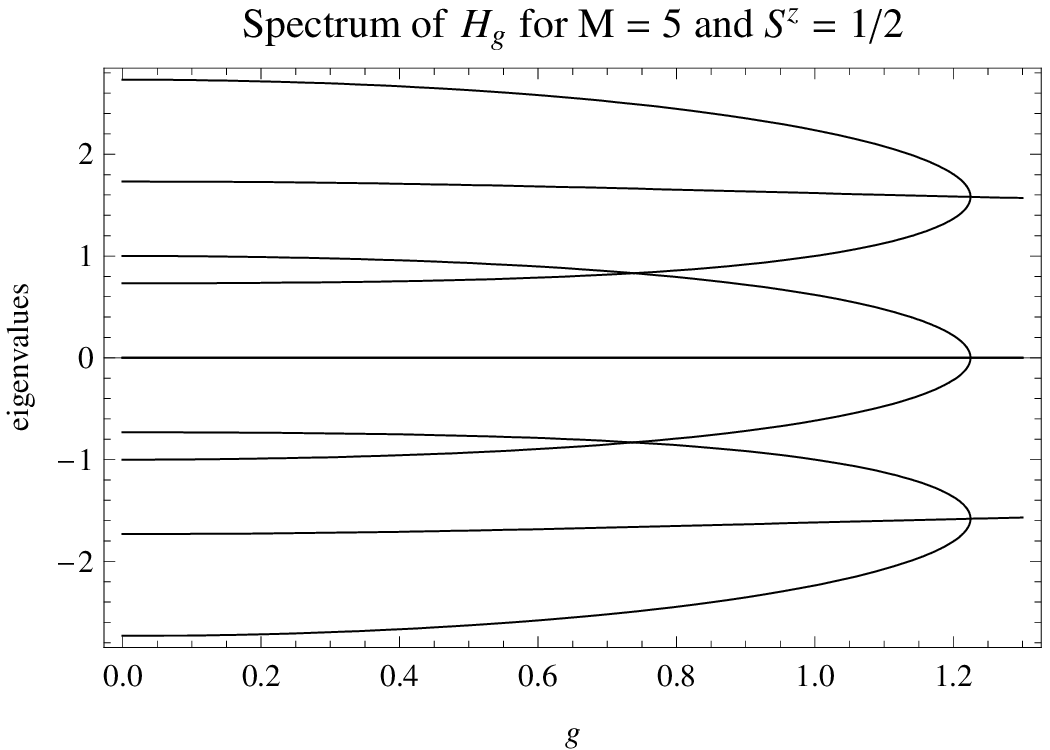}
\end{equation*}%
{\small \medskip }

\begin{center}
{\small Figure 4. Spectrum of the Hamiltonian (\ref{Hg}).\medskip }
\end{center}

\noindent Beyond the threshold value $g^{2}=\frac{M+1}{M-1}$ the Hamiltonian
has complex eigenvalues. Thus, quasi-Hermiticty of the Hamiltonian can only
hold for $g^{2}<\frac{M+1}{M-1}$ and this value approaches the unit circle
as $M\rightarrow \infty $.

\subsection{Even number of sites}

For $M$ even we define the operators $U,V$ analogously to (\ref{U}) and (\ref%
{V}), and obtain again representations of $\mathfrak{gl}(1|1)$. But we now
have%
\begin{equation}
\lbrack U,U^{\ast }]_{+}=[V,V^{\ast }]_{+}=\frac{M}{2}~\mathbf{1}\qquad 
\text{and\qquad }Z\mapsto \frac{M}{2}(1-g^{2})~\mathbf{1}\;.
\end{equation}%
The definition of the other operators in (\ref{gl11rep}) is unchanged. Due
to the absence of a zero mode, we do not obtain as before a $U(\mathfrak{gl}%
(1|1))$-symmetry of the \emph{full} Hamiltonian. Instead the generators $%
X^{\pm }$ only commute with the \emph{truncated} Hamiltonian where the last
Temperley-Lieb generator is omitted,%
\begin{equation}
\lbrack H_{g}^{\text{trunc}},U(\mathfrak{gl}(1|1))]=0
\end{equation}%
with%
\begin{equation}
H_{g}^{\text{trunc}}=H_{g}-e_{M-1}=\sum_{x=1}^{M-2}e_{x}=\frac{1}{2}%
\sum_{x=1}^{M-2}\left[ \sigma _{m}^{x}\sigma _{m+1}^{x}+\sigma
_{m}^{y}\sigma _{m+1}^{y}+ig(\sigma _{m}^{z}-\sigma _{m+1}^{z})\right] \;.
\end{equation}%
This already signals that this representation mimics the one of the chain
with $M-1$ sites. There is more evidence for this interpretation: for an
extension of the $U_{q}(sl_{2})$-representation one now needs to introduce
the operators%
\begin{equation}
M\in 2\mathbb{N}:\quad K^{\pm 1}=i^{\pm M\mp 1}\prod_{x=1}^{M}\sigma
_{x}^{z},\qquad E=\frac{U^{\ast }-igV^{\ast }}{\sqrt{\frac{M}{2}(1-g^{2})}}%
,\qquad F=-\frac{(gV+iU)K}{\sqrt{\frac{M}{2}(1-g^{2})}},
\end{equation}%
which again satisfy the right commutation relations when setting $q=\sqrt{-1}
$. Note, however, that the limit $g\rightarrow 1$ is now ill-defined and
that we have replaced $M\rightarrow M-1$ in the Cartan generator $K$. As for 
$M$ odd we have the relations%
\begin{equation}
0\leq g\leq 1:\quad \lbrack E,e_{x}+e_{x+1}]=[F,e_{x}+e_{x+1}]=0,\qquad
x=1,3,5,...,M-3\;.  \label{modSW}
\end{equation}

An alternative way to restore quantum group or $U(\mathfrak{gl}(1|1))$%
-symmetry is by adding for $M/2$ even the generator%
\begin{equation}
e_{M}=c_{M}c_{1}^{\ast }-c_{M}^{\ast }c_{1}+ig(n_{M}-n_{1})
\end{equation}%
and consider the Hamiltonian with periodic boundary conditions%
\begin{equation}
H_{\text{periodic}}=\sum_{x=1}^{M}e_{x}=\sum_{x=1}^{M}(c_{x}c_{x+1}^{\ast
}-c_{x}^{\ast }c_{x+1}),\qquad M+1\equiv 1\;.
\end{equation}%
The latter apparently does not depend on $g$ and is simply the Hermitian
Hamiltonian which describes free fermions on a lattice. We will not discuss
this Hamiltonian further in the present context.

Finally, for $M$ even there is an alternative representation of the quantum
group symmetry which is obtained from $PT$-reversal. Noting the
transformation identities%
\begin{equation}
PUP=\left\{ 
\begin{array}{cc}
(-)^{\frac{M-1}{2}}U, & M\text{ odd} \\ 
(-)^{\frac{M}{2}}V, & M\text{ even}%
\end{array}%
\right. \qquad \text{and}\qquad PVP=\left\{ 
\begin{array}{cc}
(-)^{\frac{M-1}{2}}V, & M\text{ odd} \\ 
(-)^{\frac{M}{2}}U, & M\text{ even}%
\end{array}%
\right.
\end{equation}%
one easily verifies that%
\begin{equation}
K^{\pm 1}\mapsto i^{\mp M\pm 1}\prod_{x=1}^{M}\sigma _{x}^{z},\qquad
E\mapsto \frac{V^{\ast }+igU^{\ast }}{\sqrt{\frac{M}{2}(1-g^{2})}},\qquad
F\mapsto -\frac{(gU-iV)K}{\sqrt{\frac{M}{2}(1-g^{2})}}
\end{equation}%
yields another representation of $U_{q}(sl_{2})$ with $q=-\sqrt{-1}$. The
commutation relations (\ref{modSW}) are then modified accordingly by
conjugation with the $PT$ operator.

Let us return to the case $g=1$ for an even number of sites. If we only
consider the fermionic oscillator modes $X^{\pm }$ as before for $M$ odd,
then we see that they anticommute at $g=1$, i.e. we have again $Z=0$. We
encounter the same scenario as before for $M$ odd, the corresponding
discrete wave function has zero norm,%
\begin{equation*}
M\in 2\mathbb{N}:\quad \sum_{x=1}^{M}\psi _{k=i}^{2}(x)=0\;,
\end{equation*}%
and the Hamiltonian has non-trivial Jordan blocks, however these are now of
size $2\times 2$. To be concrete we state here the Jordan normal form $J$ of
the Hamiltonian for $M=4$ when restricted to the $S^{z}=0$ subspace,%
\begin{equation*}
J|_{S^{z}=0}=\left( 
\begin{array}{cccccc}
0 & 0 & 0 & 0 & 0 & 0 \\ 
0 & 0 & 0 & 0 & 0 & 0 \\ 
0 & 0 & -\sqrt{2} & 1 & 0 & 0 \\ 
0 & 0 & 0 & -\sqrt{2} & 0 & 0 \\ 
0 & 0 & 0 & 0 & \sqrt{2} & 1 \\ 
0 & 0 & 0 & 0 & 0 & \sqrt{2}%
\end{array}%
\right) \ .
\end{equation*}%
Looking at the spectrum of the Hamiltonian we see once more that some
eigenvalues smoothly join up at $g=1$, compare with Figure 5.%
\begin{equation*}
\includegraphics[scale=0.9]{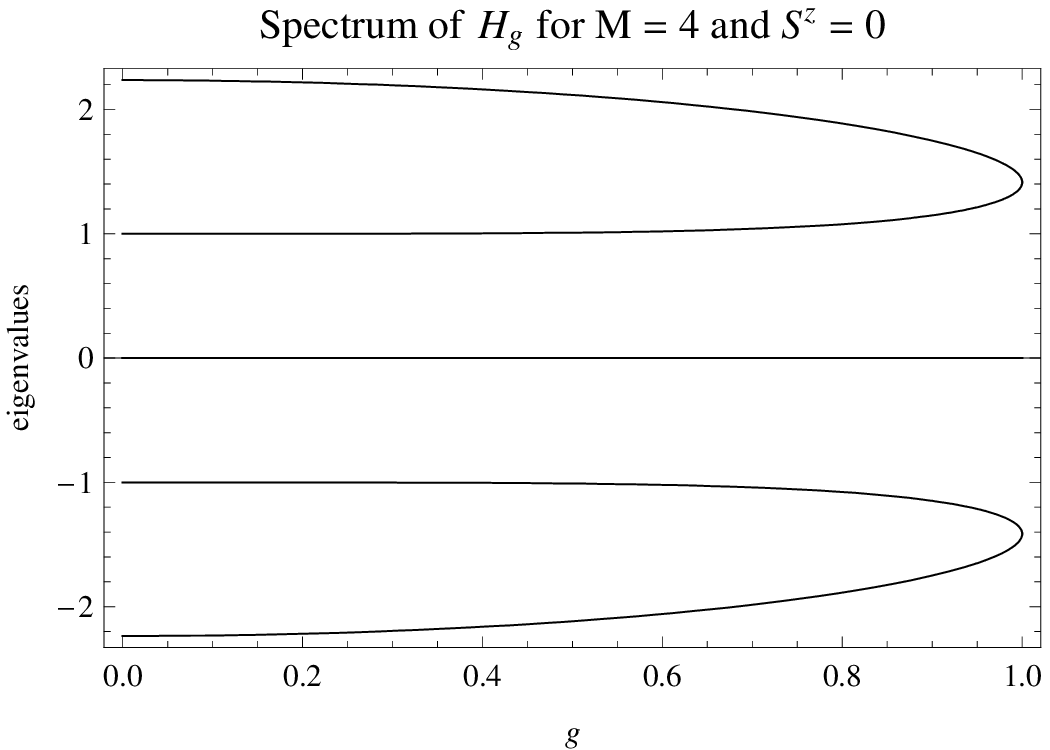}
\end{equation*}%
{\small \medskip }

\begin{center}
{\small Figure 5. Spectrum of the Hamiltonian (\ref{Hg}).\medskip }
\end{center}

\noindent This signals the appearance of complex eigenvalues for $|g|>1$.
Thus, quasi-Hermiticity only holds true inside the unit disc for $M$ even.

\section{Perturbation theory on the imaginary axis}

While we have shown the existence of the map $\eta $ for boundary fields $%
\alpha =\bar{\beta}$ inside the unit disc, its expression in terms of
projectors on quasi-fermion states does not allow to easily compute the
corresponding Hermitian Hamiltonian%
\begin{equation*}
h=\eta ^{1/2}H\eta ^{-1/2}\;.
\end{equation*}%
We therefore specialize now to the case $0\leq g<1$ and $\theta =\pi /2$ in (%
\ref{boundpar}) and compute it perturbatively. We will arrive at the result
already presented in the introduction: the presence of non-Hermitian
boundary terms leads to a long range interaction in the bulk.

Since $\eta >0$ we can write it in exponential form $\eta =\exp A$. The
quasi-Hermiticity relation then translates into%
\begin{equation}
H^{\ast }=e^{A}He^{-A}\;.
\end{equation}%
Employing the Baker-Campbell-Hausdorff formula the above identity can be
expanded into multiple commutators%
\begin{equation}
H+\sum_{k>0}\frac{1}{k!}\limfunc{ad}\nolimits_{A}^{k}H=H^{\ast },\qquad 
\limfunc{ad}\nolimits_{X}Y=[X,Y],\quad \eta =e^{A}\;.
\end{equation}%
This allows to determine the operator $A$ perturbatively by separating the
non-Hermitian Hamiltonian $H$ into an Hermitian and non-Hermitian part of
the following form,%
\begin{equation}
H=H_{0}+igH_{1},\quad \quad H_{0}=H_{0}^{\ast },\;\quad H_{1}=H_{1}^{\ast
},\quad \;0\leq g\ll 1\;.
\end{equation}%
The Hamiltonian (\ref{Hg}) has precisely this structure. Setting the
boundary parameters in (\ref{H}) to purely imaginary values, $\alpha =-\beta
=ig,\ g\in \mathbb{R}_{+}$, we define\ 
\begin{equation}
H_{0}=-\sum_{x=1}^{M-1}\left[ c_{x}^{\ast }c_{x+1}-c_{x}c_{x+1}^{\ast }%
\right] ,\qquad H_{1}=\frac{\sigma _{1}^{z}-\sigma _{M}^{z}}{2}=n_{1}-n_{M}
\end{equation}%
Since we have $H_{g}^{\ast }=H_{-g}$ we demand $\eta _{-g}=\eta _{g}^{-1}$
leading to the following ansatz for a power series expansion of the metric
operator in the coupling constant $g$,%
\begin{equation}
A=\sum_{n\geq 0}g^{2n+1}A_{2n+1}\ .  \label{Aexpansion}
\end{equation}%
From our previous considerations related to the appearance of Jordan blocks
when the central element $Z\in \mathfrak{gl}(1|1)$ vanishes, we infer that
the region of convergence for the above series is $0\leq g<1$ for $M$ even
and $0\leq g<\sqrt{\left( M+1\right) /\left( M-1\right) }$ for $M$ odd.

\begin{theorem}
Collecting terms order by order in $g$ we obtain the following equations for
the coefficients in the expansion of the operator $A$,%
\begin{equation}
\lbrack H_{0},A_{1}]=2iH_{1}  \label{A1}
\end{equation}%
and for $n\geq 1$%
\begin{equation}
\lbrack H_{0},A_{2n+1}]=i\sum_{k=1}^{n}\lambda
_{k}\sum_{p_{1}+...+p_{2k}=2n}[A_{p_{1}},...,[A_{p_{2k}},H_{1}]...]\;.
\label{Aeqns}
\end{equation}%
Here the sum runs over all \emph{compositions} $p=(p_{1},...,p_{2k})$ of $2n$
with $A_{p_{i}}=0$ if $p_{i}$ is even. The coefficients $\lambda _{k}$ are
determined recursively via%
\begin{equation}
\lambda _{1}=1/6\quad \text{and}\quad \quad \lambda _{k}=\frac{2k-1}{(2k+1)!}%
-\sum_{j=1}^{k-1}\frac{\lambda _{k-j}}{(2j+1)!}~,\;k>1\ .
\label{coefficients}
\end{equation}%
The first elements in this sequence are%
\begin{equation*}
\lambda =(\frac{1}{6},-\frac{1}{360},\frac{1}{15120},-\frac{1}{604800},\frac{%
1}{23950080},-\frac{691}{653837184000},\frac{1}{37362124800},...)
\end{equation*}
\end{theorem}

\begin{proof}
The recursion formula for the coefficients $\lambda _{k}$ is proved via
induction as follows. The case $n=1$ is easily verified. Now assume that (%
\ref{Aeqns}) holds true for all orders $k=1,2,...,n$. From the
Baker-Campbell-Hausdorff formula%
\begin{equation*}
-2iH_{1}=\sum_{k>0}\frac{1}{k!}\limfunc{ad}\nolimits_{A}^{k}H
\end{equation*}%
we find that%
\begin{eqnarray*}
\lbrack H_{0},A_{2n+3}] &=&i\sum_{k=1}^{n+1}\frac{1}{(2k)!}%
\sum_{p_{1}+...+p_{2k}=2n+2}[A_{p_{1}},...,[A_{p_{2k}},H_{1}]...] \\
&&+\sum_{k=1}^{n+1}\frac{1}{(2k+1)!}%
\sum_{p_{1}+...+p_{2k+1}=2n+3}[A_{p_{1}},...,[A_{p_{2k+1}},H_{0}]...]\ .
\end{eqnarray*}%
Employing our assumption we can rewrite the second term on the right hand
side in terms of $H_{1}$. For $p_{2k+1}=1$ we find
\begin{equation*}
[A_{p_{2k+1}},H_{0}]=-2iH_{1}
\end{equation*}
and for $1<p_{2k+1}\leq 2(n+1-k)+1$ we have
\begin{equation*}
[A_{p_{2k+1}},H_{0}]=-i\sum_{j=1}^{\frac{%
p_{2k+1}-1}{2}}\lambda
_{j}\sum_{q_{1}+...+q_{2j}=p_{2k+1}-1}[A_{q_{1}},...,[A_{q_{2j}},H_{1}]...]\
.
\end{equation*}%
Inserting these expressions, rearranging sums and collecting terms we find%
\begin{equation*}
\lbrack H_{0},A_{2n+3}]=i\sum_{k=1}^{n+1}\left\{ \frac{1}{(2k)!}-\frac{2}{%
(2k+1)!}-\tsum_{j=1}^{k-1}\frac{\lambda _{k-j}}{(2j+1)!}\right\}
\sum_{p_{1}+...+p_{2k}=2n+2}[A_{p_{1}},...,[A_{p_{2k}},H_{1}]...]
\end{equation*}%
which yields the desired recursion formula for the coefficients.
\end{proof}\medskip

In order to facilitate the comparison with results in the literature, we
explicitly state the identities up to order seven:%
\begin{eqnarray}
\lbrack H_{0},A_{1}] &=&2iH_{1}  \notag \\
\lbrack H_{0},A_{3}] &=&\frac{i}{6}[A_{1},[A_{1},H_{1}]]  \notag \\
\lbrack H_{0},A_{5}] &=&\frac{i}{6}[A_{1},[A_{3},H_{1}]]+\frac{i}{6}%
[A_{3},[A_{1},H_{1}]]-\frac{i}{360}[A_{1},[A_{1},[A_{1},[A_{1},H_{1}]]]]
\end{eqnarray}%
and%
\begin{eqnarray*}
\lbrack H_{0},A_{7}] &=&\frac{i}{6}[A_{3},[A_{3},H_{1}]]+\frac{i}{6}%
[A_{1},[A_{5},H_{1}]]+\frac{i}{6}[A_{5},[A_{1},H_{1}]] \\
&&-\frac{i}{360}%
\sum_{p_{1}+p_{2}+p_{3}+p_{4}=6}[A_{p_{1}},[A_{p_{2}},[A_{p_{3}},[A_{p_{4}},H_{1}]]]]
\\
&&+\frac{i}{15120}[A_{1},[A_{1},[A_{1},[A_{1},[A_{1},[A_{1},H_{1}]]]]]]
\end{eqnarray*}%
Note that our general formula differs from equation (3.20) in \cite{CFMFAF06}
but reproduces the identities in equation (34) of \cite{BBJ04}.

One can now use the identities (\ref{Aeqns}) to compute the metric operator $%
A$ order by order. In addition to the latter equations one also easily
deduces from the properties $\eta $ ought to obey that we must have%
\begin{equation}
\lbrack A,S^{z}]=0,\text{\qquad }PA=-AP\text{\qquad and\qquad }A^{\ast }=A\ .
\end{equation}%
Define 
\begin{equation}
a_{x,y}^{\pm }=c_{x}^{\ast }c_{y}\pm c_{x}c_{y}^{\ast },\qquad x<y
\label{axy}
\end{equation}%
then we find that up to order 11 the terms in the expansion of the matrix $A$
are of the general form%
\begin{equation}
A_{2n+1}=\frac{(-1)^{n}i}{(2n+1)}\sum_{x=1}^{M-2n-1}a_{x,x+2n+1}^{+}+%
\sum_{p=0}^{n-1}\sum_{x=1}^{n-p}i\kappa
_{x}^{(n,p)}(a_{x,x+2p+1}^{+}+a_{M-x-2p,M+1-x}^{+})\ .  \label{Aform}
\end{equation}%
Explicitly they read,%
\begin{eqnarray*}
A_{1} &=&i\sum_{x=1}^{M-1}a_{x,x+1}^{+} \\
A_{3} &=&\frac{1}{3i}\sum_{x=1}^{M-3}a_{x,x+3}^{+}-\frac{1}{6i}\left(
a_{1,2}^{+}+a_{M-1,M}^{+}\right)
\end{eqnarray*}%
\begin{equation*}
A_{5}=\frac{i}{5}\sum_{x=1}^{M-5}a_{x,x+5}^{+}+\frac{i}{24}\left(
a_{1,2}^{+}+a_{M-1,M}^{+}\right) +\frac{i}{120}\left(
a_{2,3}^{+}+a_{M-2,M-1}^{+}\right) -\frac{11i}{120}\left(
a_{1,4}^{+}+a_{M-3,M}^{+}\right) \ 
\end{equation*}%
\begin{eqnarray*}
A_{7} &=&\frac{1}{7i}\sum_{x=1}^{M-7}a_{x,x+7}^{+} \\
&&-\frac{7}{240i}\left( a_{1,2}^{+}+a_{M-1,M}^{+}\right) -\frac{1}{48i}%
\left( a_{2,3}^{+}+a_{M-2,M-1}^{+}\right) -\frac{13}{840i}\left(
a_{3,4}^{+}+a_{M-3,M-2}^{+}\right) \\
&&+\frac{1}{60i}\left( a_{1,4}^{+}+a_{M-3,M}^{+}\right) -\frac{3}{560i}%
\left( a_{2,5}^{+}+a_{M-4,M-1}^{+}\right) -\frac{103}{1680i}\left(
a_{1,6}^{+}+a_{M-5,M}^{+}\right) \ .
\end{eqnarray*}%
and%
\begin{multline*}
A_{9}=\frac{i}{9}\sum_{x=1}^{M-9}a_{x,x+9}^{-}+\tfrac{i}{64}\left(
a_{1,2}^{-}+a_{M-1,M}^{-}\right) +\tfrac{23i}{2240}\left(
a_{2,3}^{-}+a_{M-2,M-1}^{-}\right) \\
+\tfrac{17i}{1920}\left( a_{3,4}^{-}+a_{M-3,M-2}^{-}\right) +\tfrac{25i}{8064%
}\left( a_{4,5}^{-}+a_{M-4,M-3}^{-}\right) +\tfrac{11}{560i}\left(
a_{1,4}^{-}+a_{M-3,M}^{-}\right) \\
+\tfrac{29}{1920i}\left( a_{2,5}^{-}+a_{M-4,M-1}^{-}\right) +\tfrac{587}{%
40320i}\left( a_{3,6}^{-}+a_{M-5,M-2}^{-}\right) + \\
+\tfrac{113i}{13440}\left( a_{1,6}^{-}+a_{M-5,M}^{-}\right) +\tfrac{59}{8064i%
}\left( a_{2,7}^{-}+a_{M-6,M-1}^{-}\right) +\tfrac{1823}{40320i}\left(
a_{1,8}^{-}+a_{M-7,M}^{-}\right)
\end{multline*}%
The stated solutions can be checked by employing the relations%
\begin{eqnarray*}
\lbrack H_{0},a_{x,y}^{\pm }] &=&a_{x,y+1}^{\mp }-a_{x-1,y}^{\mp
}+a_{x,y-1}^{\mp }-a_{x+1,y}^{\mp } \\
\lbrack H_{1},a_{x,y}^{\pm }] &=&(\delta _{x,1}+\delta _{y,M})a_{x,y}^{\mp
},\qquad x<y
\end{eqnarray*}%
and%
\begin{equation*}
\lbrack a_{x,y}^{+},a_{r,s}^{-}]=\delta _{y,r}a_{x,s}^{-}+\delta
_{y,s}a_{x,r}^{-}-\delta _{x,s}a_{r,y}^{-}-\delta _{x,r}a_{s,y}^{-}\ .
\end{equation*}%
Note that some of the terms in the general formula (\ref{Aform}) vanish if $%
2n+1>M$. The general solution still needs to be found. Once the metric
operator $A$ is obtained, the corresponding Hermitian Hamiltonian $h$ can be
computed order by order. Define the Hermitian Hamiltonian $h$ according to%
\begin{equation}
h=e^{\frac{A}{2}}He^{-\frac{A}{2}}=H_{0}+\sum_{n=1}^{\infty }g^{2n}h_{2n}\;,
\end{equation}%
where the same region of convergence is implied as in the case of $A$.

\begin{theorem}
The terms in the series expansion of $h$ are given by%
\begin{equation}
h_{2n}=i\sum_{k=1}^{n}\lambda _{k}^{\prime
}\sum_{p_{1}+...+p_{2k-1}=2n-1}[A_{p_{1}},...,[A_{p_{k}},H_{1}]...]\;,
\label{hcoeff}
\end{equation}%
where the sum runs over all \emph{compositions} $p=(p_{1},...,p_{2k-1})$ of $%
2n-1$ and the coefficients $\lambda _{k}^{\prime }$ are computed from the
coefficients $\lambda _{k}$ in (\ref{Aeqns}) via the formula%
\begin{equation}
\lambda _{k}^{\prime }=\frac{2k-1}{2^{2k-1}(2k)!}-\sum_{j=1}^{k-1}\frac{%
\lambda _{k-j}}{2^{2j}(2j)!}\ .
\end{equation}%
The first terms in the above sequence are%
\begin{equation*}
\lambda ^{\prime }=(\frac{1}{4},-\frac{1}{192},\frac{1}{7680},-\frac{17}{%
5160960},\frac{31}{371589120},\frac{691}{326998425600},...)
\end{equation*}
\end{theorem}

\begin{proof}
An induction proof similar to the previous one for the series expansion of
the operator $A$.
\end{proof}\medskip

Again, we state the first terms explicitly in order to allow for comparison
with the literature (see equation (3.24) in \cite{CFMFAF06}),%
\begin{eqnarray*}
h_{2} &=&\frac{i}{4}~[A_{1},H_{1}] \\
h_{4} &=&\frac{i}{4}~[A_{3},H_{1}]-\frac{i}{192}%
~[A_{1},[A_{1},[A_{1},H_{1}]]] \\
h_{6} &=&\frac{i}{4}~[A_{5},H_{1}]+\frac{i}{192}%
\sum_{p_{1}+p_{2}+p_{3}=5}[A_{p_{1}},[A_{p_{2}},[A_{p_{3}},H_{1}]]] \\
&&+\frac{i}{7680}[A_{1},[A_{1},[A_{1},[A_{1},[A_{1},H_{1}]]]]]\;.
\end{eqnarray*}%
Note that $h$ only depends on $g^{2}$ as it must due to $H_{g}^{\ast
}=H_{-g} $. Inserting the explicit expressions for the series expansion of $%
A $ up to order $n=5$ we find the Hamiltonian stated in the introduction.

\section{Exact results for short spin-chains}

In this section we present exact expressions for the quasi-Hermiticity
operator $\eta $ and the operator $A$ employing the fact that for small $M$
there are only a few terms appearing in the expressions for the operator $A$%
, compare with (\ref{Aform}). That is, motivated by our perturbation theory
results we make the ansatz%
\begin{equation}
A=\sum_{p=0}^{\left\lfloor M/2-1\right\rfloor }\sum_{x=1}^{M-2p-1}\varkappa
_{x}^{(p)}(a_{x,x+2p+1}^{+}+a_{M-x-2p,M+1-x}^{+})\ .  \label{Aansatz}
\end{equation}%
For $M$ small enough one can exponentiate the resulting matrix and solve the
intertwining relation in order to determine the coefficients $\varkappa
_{x}^{(p)}$. From our perturbative computation we can also make an educated
guess about the form of the Hermitian Hamiltonian $h$ in order to facilitate
the computation of the similarity transformation. Namely, we make the ansatz%
\begin{equation}
h=\sum_{p=0}^{\left\lfloor M/2-1\right\rfloor }\sum_{x=1}^{M-2p-1}\xi
_{x}^{(p)}(a_{x,x+2p+1}^{-}+a_{M-x-2p,M+1-x}^{-})\;.  \label{hansatz}
\end{equation}%
Here we have used that $h$ is parity invariant. This follows from the
identities%
\begin{equation*}
P\eta ^{\frac{1}{2}}=Pe^{A/2}=e^{-A/2}P=\eta ^{-\frac{1}{2}}P
\end{equation*}%
and%
\begin{equation}
Ph=P\eta ^{\frac{1}{2}}H\eta ^{-\frac{1}{2}}=\eta ^{-\frac{1}{2}}H^{\ast
}\eta ^{\frac{1}{2}}P=hP\;.
\end{equation}

In the last part of this section we present another exact approach which is
valid only at $g=1$ for $M$ odd but it exploits the graphical calculus
attached to the Temperley-Lieb algebra and is interesting in its own right.

\subsection{Example $M=3$}

It might be instructive to start out with the easiest, albeit not the most
interesting, case $M=3$. The non-Hermitian Hamiltonian reads%
\begin{equation}
H=-a_{1,2}^{-}-a_{2,3}^{-}+ig(n_{1}-n_{3})
\end{equation}%
where $a_{x,y}^{-}$ has been defined in (\ref{axy}). The operator $A=\ln
\eta $ is then computed to be%
\begin{equation}
A=-\frac{\arccos \frac{2+g^{2}}{2-g^{2}}}{\sqrt{2}}%
~(a_{1,2}^{+}+a_{2,3}^{+})\;.
\end{equation}%
After exponentiation we obtain the quasi-Hermiticity operator and from its
explicit form,%
\begin{equation}
\eta =\left( 
\begin{array}{llllllll}
1 & 0 & 0 & 0 & 0 & 0 & 0 & 0 \\ 
0 & -\frac{2}{g^{2}-2} & -\frac{2ig}{g^{2}-2} & \frac{g^{2}}{g^{2}-2} & 0 & 0
& 0 & 0 \\ 
0 & \frac{2ig}{g^{2}-2} & -1-\frac{4}{g^{2}-2} & -\frac{2ig}{g^{2}-2} & 0 & 0
& 0 & 0 \\ 
0 & \frac{g^{2}}{g^{2}-2} & \frac{2ig}{g^{2}-2} & -\frac{2}{g^{2}-2} & 0 & 0
& 0 & 0 \\ 
0 & 0 & 0 & 0 & -\frac{2}{g^{2}-2} & -\frac{2ig}{g^{2}-2} & \frac{g^{2}}{%
g^{2}-2} & 0 \\ 
0 & 0 & 0 & 0 & \frac{2ig}{g^{2}-2} & -1-\frac{4}{g^{2}-2} & -\frac{2ig}{%
g^{2}-2} & 0 \\ 
0 & 0 & 0 & 0 & \frac{g^{2}}{g^{2}-2} & \frac{2ig}{g^{2}-2} & -\frac{2}{%
g^{2}-2} & 0 \\ 
0 & 0 & 0 & 0 & 0 & 0 & 0 & 1%
\end{array}%
\right) ,
\end{equation}%
one can explicitly verify that $\eta $ satisfies all the necessary
conditions. Moreover, as predicted from our algebraic analysis related to
the $U(\mathfrak{gl}(1|1))$-symmetry of the Hamiltonian, we see that $\eta $
has a pole precisely at the value $g^{2}=(M+1)/(M-1)=2$ where the central
element $Z\in U(\mathfrak{gl}(1|1))$ vanishes. The renormalized $\eta $,%
\begin{equation*}
\eta \rightarrow (g^{2}-2)\eta ^{\prime },
\end{equation*}%
still intertwines the Hamiltonian, $\eta ^{\prime }H=H^{\ast }\eta ^{\prime
} $ at $g^{2}=2$, but is not any longer positive definite but semi-definite.
In principle one can also compute the square root for $\eta ^{1/2}$ for
general values of $g$ but as the expressions become rather unwieldy we
restrict ourselves to stating the result for $g=1,$%
\begin{equation}
\lim_{g\rightarrow 1}\eta ^{\frac{1}{2}}=\left( 
\begin{array}{llllllll}
1 & 0 & 0 & 0 & 0 & 0 & 0 & 0 \\ 
0 & \frac{1}{2}+\frac{1}{\sqrt{2}} & \frac{i}{\sqrt{2}} & \frac{1}{2}-\frac{1%
}{\sqrt{2}} & 0 & 0 & 0 & 0 \\ 
0 & -\frac{i}{\sqrt{2}} & \sqrt{2} & \frac{i}{\sqrt{2}} & 0 & 0 & 0 & 0 \\ 
0 & \frac{1}{2}-\frac{1}{\sqrt{2}} & -\frac{i}{\sqrt{2}} & \frac{1}{2}+\frac{%
1}{\sqrt{2}} & 0 & 0 & 0 & 0 \\ 
0 & 0 & 0 & 0 & \frac{1}{2}+\frac{1}{\sqrt{2}} & \frac{i}{\sqrt{2}} & \frac{1%
}{2}-\frac{1}{\sqrt{2}} & 0 \\ 
0 & 0 & 0 & 0 & -\frac{i}{\sqrt{2}} & \sqrt{2} & \frac{i}{\sqrt{2}} & 0 \\ 
0 & 0 & 0 & 0 & \frac{1}{2}-\frac{1}{\sqrt{2}} & -\frac{i}{\sqrt{2}} & \frac{%
1}{2}+\frac{1}{\sqrt{2}} & 0 \\ 
0 & 0 & 0 & 0 & 0 & 0 & 0 & 1%
\end{array}%
\right) \;.
\end{equation}%
After performing the similarity transformation we end up with the Hermitian
Hamiltonian%
\begin{equation}
\lim_{g\rightarrow 1}h=\frac{a_{1,2}^{-}+a_{2,3}^{-}}{\sqrt{2}}\;.
\end{equation}%
Arguably, we do not see here the manifestation of the long-range nature of
the interaction, since the spin-chain is simply too short. However, this
case illustrates the working of the general ideas and concepts introduced
earlier.

\subsection{Example $M=4$}

We saw previously that the cases $M$ even and odd are different and it is
therefore important to have another simple but concrete example. We now find 
\begin{equation}
A=i\frac{\xi \sqrt{1+\omega ^{2}}}{2}(a_{1,2}^{+}+a_{3,4}^{+})+\frac{\xi
\zeta +\omega }{2}~a_{2,3}^{+}+\frac{\xi \zeta -\omega }{2}~a_{1,4}^{+}
\end{equation}%
where the coefficients are given by%
\begin{equation}
\xi =2\func{arctanh}\frac{g\sqrt{5-g^{2}}}{1+g^{2}+\sqrt{1-g^{2}}},\qquad
\zeta =2i\func{arctanh}\frac{-1+\sqrt{1-g^{2}}}{g}\qquad \omega =i\sqrt{%
\frac{1-g^{2}}{5-g^{2}}}\;.
\end{equation}%
Inserting these expressions into the formula for $A$ does not lead to any
further drastic simplifications. Also the expression for the
quasi-Hermiticity operator in the spin-sector $S^{z}=0$ looks now more
complicated: after rescaling%
\begin{equation*}
\gamma =\sqrt{1-g^{2}},\qquad \eta ^{\prime }=\gamma (3-g^{2}-2\gamma )\eta
\end{equation*}%
one finds the matrix%
\begin{equation}
\eta ^{\prime }=\left( 
\begin{array}{cccccc}
1 & ig & 1-g^{2}-\gamma & 1-g^{2}-\gamma & ig\left( 2-g^{2}-2\gamma \right)
& g^{2}+2\gamma -2 \\ 
\ast & 3-2\gamma & -ig\left( \gamma -2\right) & -ig\left( \gamma -2\right) & 
-g^{2} & ig\left( 2-g^{2}-2\gamma \right) \\ 
\ast & \ast & \gamma -g^{2}\left( \gamma -2\right) & 2-2\gamma & -ig\left(
\gamma -2\right) & 1-g^{2}-\gamma \\ 
\ast & \ast & \ast & \gamma -g^{2}\left( \gamma -2\right) & -ig\left( \gamma
-2\right) & 1-g^{2}-\gamma \\ 
\ast & \ast & \ast & \ast & 3-2\gamma & ig \\ 
\ast & \ast & \ast & \ast & \ast & 1%
\end{array}%
\right) \ .
\end{equation}%
Here we have omitted the lower diagonal part, since $\eta ^{t}=\bar{\eta}$.
Again one can explicitly check that this operator satisfies all the
necessary requirements. From our previous discussion we expect that in the
limit $g\rightarrow 1$ the operator $\eta ^{\prime }$ should cease to be
positive definite and become semi-definite. In fact, one verifies that this
is the case by using the expression%
\begin{equation}
\lim_{g\rightarrow 1}\eta ^{\prime }=\left( 
\begin{array}{llllll}
1 & i & 0 & 0 & i & -1 \\ 
-i & 3 & 2i & 2i & -1 & i \\ 
0 & -2i & 2 & 2 & 2i & 0 \\ 
0 & -2i & 2 & 2 & 2i & 0 \\ 
-i & -1 & -2i & -2i & 3 & i \\ 
-1 & -i & 0 & 0 & -i & 1%
\end{array}%
\right) \;.
\end{equation}

\subsection{Example $M=5$}

Finally to present one example where the long range nature at $g=1$ is
slightly more apparent we consider the $M=5$ chain. We specialize from the
start to the case $g=1$ and find that%
\begin{equation}
\lim_{g\rightarrow 1}A=\frac{\xi +\zeta \omega }{2}(a_{1,2}^{-}+a_{4,5}^{-})+%
\frac{\zeta \sqrt{1-\omega ^{2}}}{\sqrt{2}}(a_{2,3}^{-}+a_{3,4}^{-})+\frac{%
\xi -\zeta \omega }{2}(a_{1,4}^{-}+a_{2,5}^{-})
\end{equation}%
with%
\begin{equation*}
\xi =\frac{i\ln 5}{2},\qquad \zeta =i\ln \left[ 2+\frac{7}{\sqrt{5}}+2\sqrt{%
\tfrac{16}{5}+\tfrac{7}{\sqrt{5}}}\right] ,\qquad \omega =\sqrt{\frac{4+%
\sqrt{5}}{11}}\;.
\end{equation*}%
After exponentiation we obtain the quasi-Hermiticity operator. In the spin $%
S^{z}=3/2$ sector it reads explicitly 
\begin{equation}
\lim_{g\rightarrow 1}\eta =\left( 
\begin{array}{lllll}
1+\frac{3}{\sqrt{5}} & i+\frac{3i}{\sqrt{5}} & -1-\frac{1}{\sqrt{5}} & \frac{%
1}{i\sqrt{5}}-i & 1 \\ 
\frac{3}{i\sqrt{5}}-i & 1+\sqrt{5} & i+\frac{3i}{\sqrt{5}} & -1-\frac{2}{%
\sqrt{5}} & \frac{1}{i\sqrt{5}}-i \\ 
-1-\frac{1}{\sqrt{5}} & \frac{5+3\sqrt{5}}{5i} & 1+\frac{4}{\sqrt{5}} & i+%
\frac{3i}{\sqrt{5}} & -1-\frac{1}{\sqrt{5}} \\ 
i+\frac{i}{\sqrt{5}} & -1-\frac{2}{\sqrt{5}} & \frac{5+3\sqrt{5}}{5i} & 1+%
\sqrt{5} & i+\frac{3i}{\sqrt{5}} \\ 
1 & i+\frac{i}{\sqrt{5}} & -1-\frac{1}{\sqrt{5}} & \frac{3}{i\sqrt{5}}-i & 1+%
\frac{3}{\sqrt{5}}%
\end{array}%
\right)  \label{eta5ex}
\end{equation}%
and its square root is computed to be%
\begin{equation}
\lim_{g\rightarrow 1}\eta ^{\frac{1}{2}}=\left( 
\begin{array}{ccccc}
\frac{7+9\sqrt[4]{5}-\sqrt{5}+\frac{14}{5^{1/4}}}{22} & \frac{i\sqrt[4]{5}}{2%
} & \frac{1-5\sqrt[4]{5}+3\sqrt{5}-9/5^{1/4}}{22} & \frac{-i}{2\sqrt[4]{5}}
& \frac{7-2\sqrt[4]{5}-\sqrt{5}+\frac{3}{5^{1/4}}}{22} \\ 
-\frac{i\sqrt[4]{5}}{2} & \frac{\sqrt{4+\frac{9}{\sqrt{5}}}}{2} & \frac{i}{%
\sqrt[4]{5}} & \frac{-1}{2\sqrt[4]{5}} & -\frac{i}{2\sqrt[4]{5}} \\ 
\frac{1-5\sqrt[4]{5}+3\sqrt{5}-\frac{9}{5^{1/4}}}{22} & \frac{-i}{\sqrt[4]{5}%
} & \frac{4+2\sqrt[4]{5}+\sqrt{5}+8/5^{1/4}}{11} & \frac{i}{\sqrt[4]{5}} & 
\frac{1-5\sqrt[4]{5}+3\sqrt{5}-\frac{9}{5^{1/4}}}{22} \\ 
\frac{i}{2\sqrt[4]{5}} & \frac{-1}{2\sqrt[4]{5}} & \frac{-i}{\sqrt[4]{5}} & 
\frac{\sqrt{4+\frac{9}{\sqrt{5}}}}{2} & \frac{i\sqrt[4]{5}}{2} \\ 
\frac{7-2\sqrt[4]{5}-\sqrt{5}+\frac{3}{5^{1/4}}}{22} & \frac{i}{2\sqrt[4]{5}}
& \frac{1-5\sqrt[4]{5}+3\sqrt{5}-9/5^{1/4}}{22} & \frac{\sqrt[4]{5}}{2i} & 
\frac{7+9\sqrt[4]{5}-\sqrt{5}+\frac{14}{5^{1/4}}}{22}%
\end{array}%
\right) \;.
\end{equation}%
Using these matrix expressions one can compute the Hermitian Hamiltonian $h$
and we find that it is of the expected form,%
\begin{equation}
\lim_{g\rightarrow 1}h=\rho _{1}(a_{1,2}^{-}+a_{4,5}^{-})+\rho _{2}\left(
a_{2,3}^{-}+a_{3,4}^{-}\right) +\rho _{3}\left(
a_{1,4}^{-}+a_{2,5}^{-}\right)
\end{equation}%
with%
\begin{equation}
\rho _{1}=\tfrac{9-6\sqrt{5}-\sqrt{2\left( 15+23\sqrt{5}\right) }}{22},\quad
\rho _{2}=\tfrac{3-2\sqrt{5}-\sqrt{40+21\sqrt{5}}}{11},\quad \rho _{3}=%
\tfrac{-2+5\sqrt{5}-\sqrt{2\left( 15+23\sqrt{5}\right) }}{22}\;.
\end{equation}%
The non-local nature of the interaction is now visible in terms of the
nonvanishing coefficient $\rho _{3}$. Below we show the result for the spin
sector $S^{z}=1/2$, however we compute it via different means.

\subsection{The quasi-Hermiticity operator and the Temperley-Lieb algebra}

In this section we indicate an alternative way to compute the
quasi-Hermiticity operator $\eta $ at $g=1$ when $M$ is odd employing the
Temperley-Lieb algebra. From the expressions (\ref{fermieta}) we expect that
in general $\eta $ will be highly non-local in terms of the spin basis,
since the quasi-particle creation and annihilation operators involve sums
over the argument of the discrete wave functions. Thus, if we compute the
matrix elements of $\eta $ in a fixed spin sector then all elements will be
non-vanishing. Our example above for $M=5$ and $S^{z}=3/2$ confirms this,
compare with (\ref{eta5ex}). However, there is another choice of basis in
which $\eta $ simplifies, that is where many of its matrix elements are
zero. Namely, we consider the $q\rightarrow i=\sqrt{-1}$ limit of the dual
canonical basis \cite{Lu,FK} which we denote by $\{t_{i}\}$. The elements of
this basis are in one-to-one correspondence with algebra elements $a_{i}\in
TL_{M}$ and the latter can be written down in terms of Young tableaux, see
e.g. \cite{CS}.

Fix a spin sector, $S^{z}=const.$, and set $m=M/2-S^{z}$. Let $\lambda _{n}$
be the rectangular Young diagram with $n$ rows of $N-n$ boxes,%
\begin{equation*}
\lambda _{m}=\underset{M-m}{\underbrace{\left. 
\begin{tabular}{|l|l|l|l|l|}
\hline
&  &  &  &  \\ \hline
&  &  &  &  \\ \hline
&  &  &  &  \\ \hline
&  &  &  &  \\ \hline
\end{tabular}%
\right\} }}~m
\end{equation*}%
Then we assign to each subdiagram $\lambda ^{\prime }\subset \lambda _{m}$ a
vector as follows. Let $t$ be the unique standard tableau (column and row
strict) of shape $\lambda ^{\prime }$ whose entries are consecutive integers
with entry $n$ in the upper left corner. For example,%
\begin{equation}
t=%
\begin{tabular}{|c|cccc}
\hline
$m$ & $m+1$ & \multicolumn{1}{|c}{$m+2$} & \multicolumn{1}{|c}{$\cdots $} & 
\multicolumn{1}{|c|}{$s$} \\ \hline
$m-1$ & $m$ & \multicolumn{1}{|c}{$\cdots $} & \multicolumn{1}{|c}{$s-2$} & 
\multicolumn{1}{|c}{} \\ \cline{1-4}
$\vdots $ &  & \multicolumn{1}{|c}{} & \multicolumn{1}{|c}{} &  \\ 
\cline{1-3}
$s^{\prime }$ &  &  &  &  \\ \cline{1-1}
\end{tabular}%
~,\quad m<s<M,\quad 1\leq s^{\prime }<m\;.
\end{equation}%
Reading the entries of the tableau from left to right and top to bottom we
set%
\begin{equation}
t\mapsto e_{s^{\prime }}e_{s^{\prime }-1}\cdots e_{s-2}\cdots
e_{m-1}e_{s}\cdots e_{m+1}e_{m}\Omega _{m},
\end{equation}%
where 
\begin{equation}
\Omega _{m}=\underset{m}{\underbrace{v_{-}\otimes v_{-}\cdots \otimes v_{-}}}%
\otimes v_{+}\otimes v_{+}\cdots \otimes v_{+}
\end{equation}%
is the vector corresponding to $\lambda ^{\prime }=\varnothing $. Note that
for fixed $m$ there are as many of these tableaux as the dimension of the
spin sector, namely $\dbinom{M}{m}$. \smallskip

\noindent \textbf{Example}. Let $M=5$ and $m=2$ then we have the following
Young diagrams and tableaux:

\begin{equation*}
\Yvcentermath1t=\varnothing ,\quad \young(2),\quad \young(2,1),\quad \young%
(23),\quad \young(23,1),\quad \young(23,12),\quad \young(234),\quad \young%
(234,1),\quad \young(234,12),\quad \young(234,123)\;.
\end{equation*}%
The corresponding algebra elements $a\in TL_{M}(q)$ are 
\begin{equation*}
a=1,\;e_{2},\;e_{1}e_{2},\;e_{3}e_{2},\;e_{4}e_{3}e_{2},\;e_{1}e_{3}e_{2},%
\;e_{2}e_{1}e_{3}e_{2},\;e_{1}e_{4}e_{3}e_{2},\;e_{2}e_{1}e_{4}e_{3}e_{2},%
\;e_{3}e_{2}e_{1}e_{4}e_{3}e_{2}\ .
\end{equation*}%
Each of these algebra elements we can also represent as a link or Kauffman
diagram. Identifying%
\begin{equation*}
\includegraphics[scale=0.7]{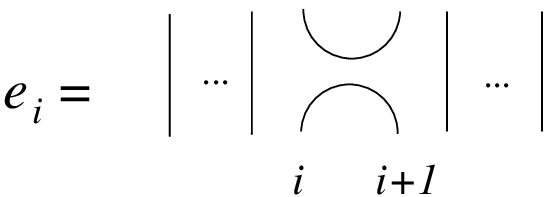}
\end{equation*}%
and realizing multiplication by concatenation from above, we find the
following diagrams for the algebra elements,

\begin{center}
\begin{tabular}{||l||c|c|c|c|c|c|}
\hline
$a\in TL_{5}$ & $a_{1}=1$ & $a_{2}=e_{2}$ & $a_{3}=e_{1}e_{2}$ & $%
a_{4}=e_{3}e_{2}$ & $a_{5}=e_{4}e_{3}e_{2}$ & $a_{6}=e_{1}e_{3}e_{2}$ \\ 
\hline
diagram & $\includegraphics[scale=0.5]{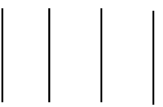}$ & $%
\includegraphics[scale=0.5]{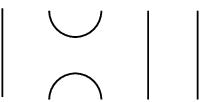}$ & %
\includegraphics[scale=0.5]{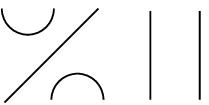} & %
\includegraphics[scale=0.5]{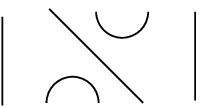} & %
\includegraphics[scale=0.5]{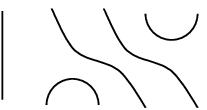} & %
\includegraphics[scale=0.5]{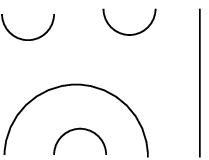} \\ \hline
\end{tabular}
\end{center}

and

\begin{center}
\begin{tabular}{||l|c|c|c|c|}
\hline
$a\in TL_{5}$ & $a_{7}=e_{2}e_{1}e_{3}e_{2}$ & $a_{8}=e_{1}e_{4}e_{3}e_{2}$
& $a_{9}=e_{2}e_{1}e_{4}e_{3}e_{2}$ & $a_{10}=e_{3}e_{2}e_{1}e_{4}e_{3}e_{2}$
\\ \hline
diagram & $\includegraphics[scale=0.5]{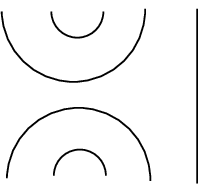}$ & %
\includegraphics[scale=0.5]{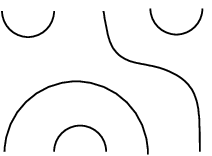} & %
\includegraphics[scale=0.5]{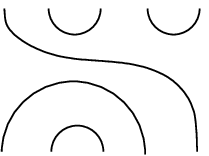} & %
\includegraphics[scale=0.5]{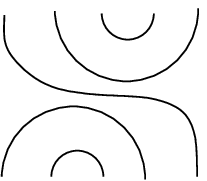} \\ \hline
\end{tabular}
\end{center}

\noindent We will make use of these diagrams momentarily. Employing the
representation $TL_{M}\rightarrow \limfunc{End}V^{\otimes M}$ introduced
earlier, 
\begin{equation*}
e_{x}\mapsto c_{x}c_{x+1}^{\ast }-c_{x}^{\ast }c_{x+1}+i(n_{x}-n_{x+1}),
\end{equation*}%
we can now generate the basis vectors $t_{i}$ by acting with each
corresponding algebra element $a_{i}$ onto the vector $\Omega _{m}$.
Computing from our previous expression (\ref{fermieta}) the
quasi-Hermiticity operator $\eta $ we can evaluate its matrix elements in
this new basis. Based on numerical computations for $M=3,5,7$ we arrive at
the following conjecture.

\begin{conjecture}
Denote by $G$ the Gram matrix of the dual canonical basis vectors $\{t_{i}\}$
with respect to the $\eta $-product, i.e.%
\begin{equation*}
G_{ij}=\langle t_{i},\eta t_{j}\rangle \ .
\end{equation*}%
Then we have%
\begin{equation*}
G_{ij}=0\qquad \text{whenever}\qquad \limfunc{tr}(a_{i}a_{j})=0\func{mod}2,
\end{equation*}%
where $a_{i},a_{j}$ are the algebra elements corresponding to $t_{i},t_{j}$
and 
\begin{equation*}
\limfunc{tr}a=\limfunc{tr}ae_{M}=\text{number of closed loops}
\end{equation*}%
which are obtained by closing the planar diagram associated with $a$. An
example for $a=e_{2}$ is shown below,%
\begin{equation*}
\includegraphics[scale=0.5]{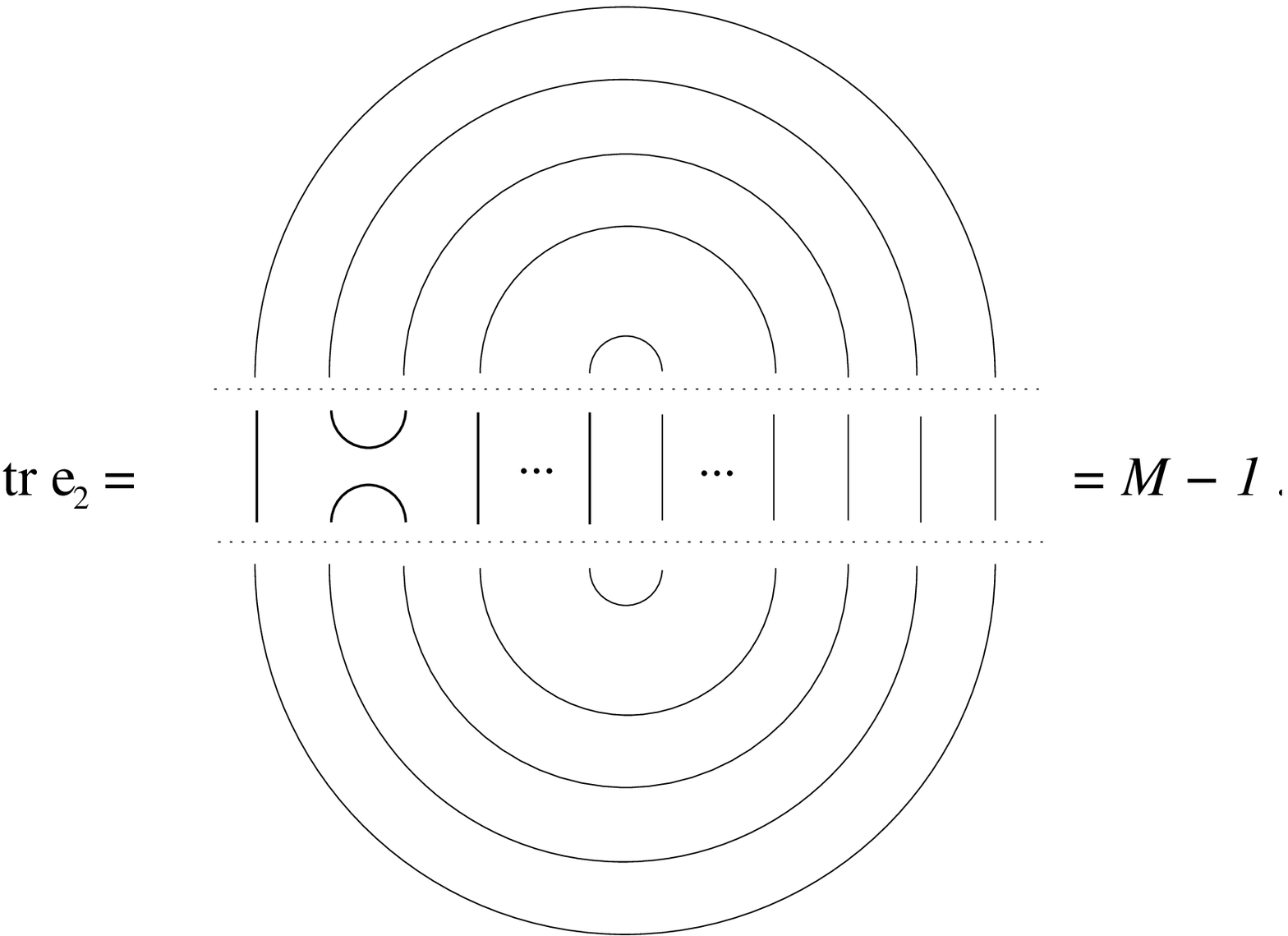}
\end{equation*}
\end{conjecture}

\noindent \textbf{Remark}. Note that the above relation does not necessarily
imply that $G_{ij}\neq 0$ if $\limfunc{tr}(a_{i}a_{j})=1\func{mod}2$. The
Gram matrix inherits from $\eta $ the further properties%
\begin{equation*}
\det G=1,\qquad G>0\qquad \text{and}\qquad G_{ij}=G_{ji}\in \mathbb{R}\ .
\end{equation*}

The knowledge of the Gram matrix is sufficient to compute matrix elements
with respect to the $\eta $-product. In particular the intertwining
property, $\eta H=H^{\ast }\eta $, translates into the following identity
for the Gram matrix%
\begin{equation}
G\mathcal{H}=\mathcal{H}^{t}G  \label{GH=HtG}
\end{equation}%
where the matrix $\mathcal{H}$ has only integer entries,\ $\mathcal{H}%
_{ji}\in \mathbb{Z}$, and is defined via 
\begin{equation}
Ht_{i}=\sum_{j}t_{j}\mathcal{H}_{ji}\;.
\end{equation}%
Note that $\mathcal{H}$\ can be computed graphically using link diagrams.
For instance, we find for our example $M=5$ and $S^{z}=1/2$ stated above that%
\begin{equation}
\mathcal{H}=\left( 
\begin{array}{llllllllll}
0 & 0 & 0 & 0 & 0 & 0 & 0 & 0 & 0 & 0 \\ 
1 & 0 & 1 & 1 & 0 & 0 & 0 & 0 & 0 & 0 \\ 
0 & 1 & 0 & 0 & 0 & 0 & 0 & 0 & 0 & 0 \\ 
0 & 1 & 0 & 0 & 1 & 0 & 0 & 0 & 0 & 0 \\ 
0 & 0 & 0 & 1 & 0 & 0 & 0 & 0 & 0 & 0 \\ 
0 & 0 & 1 & 1 & 0 & 0 & 2 & 1 & 0 & 1 \\ 
0 & 0 & 0 & 0 & 0 & 1 & 0 & 0 & 0 & 0 \\ 
0 & 0 & 0 & 0 & 1 & 1 & 0 & 0 & 1 & 0 \\ 
0 & 0 & 0 & 0 & 0 & 0 & 1 & 1 & 0 & 2 \\ 
0 & 0 & 0 & 0 & 0 & 0 & 0 & 0 & 1 & 0%
\end{array}%
\right)
\end{equation}%
The corresponding Gram matrix is%
\begin{equation}
G=\left( 
\begin{array}{llllllllll}
\frac{2\left( 3+\sqrt{5}\right) }{5} & 0 & \frac{2\left( 1+\sqrt{5}\right) }{%
5} & \frac{3}{5}+\frac{1}{\sqrt{5}} & 0 & 0 & -\frac{2}{5} & \frac{2}{5} & 0
& \frac{3}{5} \\ 
0 & 1+\frac{3}{\sqrt{5}} & 0 & 0 & 1+\frac{1}{\sqrt{5}} & 0 & 0 & 0 & 1 & 0
\\ 
\frac{2\left( 1+\sqrt{5}\right) }{5} & 0 & \frac{2\left( 2+\sqrt{5}\right) }{%
5} & \frac{1}{5}+\frac{1}{\sqrt{5}} & 0 & 0 & \frac{1}{5} & \frac{4}{5} & 0
& \frac{1}{5} \\ 
\frac{3}{5}+\frac{1}{\sqrt{5}} & 0 & \frac{1}{5}+\frac{1}{\sqrt{5}} & \frac{%
3\left( 3+\sqrt{5}\right) }{5} & 0 & 0 & \frac{4}{5} & \frac{1}{5} & 0 & 
\frac{9}{5} \\ 
0 & 1+\frac{1}{\sqrt{5}} & 0 & 0 & 1+\frac{2}{\sqrt{5}} & 0 & 0 & 0 & 1 & 0
\\ 
0 & 0 & 0 & 0 & 0 & 1 & 0 & 0 & 0 & 0 \\ 
-\frac{2}{5} & 0 & \frac{1}{5} & \frac{4}{5} & 0 & 0 & \frac{9}{5} & \frac{1%
}{5} & 0 & \frac{4}{5} \\ 
\frac{2}{5} & 0 & \frac{4}{5} & \frac{1}{5} & 0 & 0 & \frac{1}{5} & \frac{4}{%
5} & 0 & \frac{1}{5} \\ 
0 & 1 & 0 & 0 & 1 & 0 & 0 & 0 & 1 & 0 \\ 
\frac{3}{5} & 0 & \frac{1}{5} & \frac{9}{5} & 0 & 0 & \frac{4}{5} & \frac{1}{%
5} & 0 & \frac{9}{5}%
\end{array}%
\right) \ .
\end{equation}%
One can check that the above conjecture holds true for this example using
the stated link diagrams.

In addition to the identity (\ref{GH=HtG}) we obtain further constraints on
the Gram matrix by employing $PT$-invariance, $\eta ^{-1}=P\eta P=\bar{\eta}%
=\eta ^{t}$, from which we conclude that%
\begin{equation*}
G_{ij}=\langle Tt_{i},\eta ^{-1}Tt_{j}\rangle
\end{equation*}%
and hence%
\begin{equation}
\mathcal{M}^{\ast }G\mathcal{M}=G,\qquad PTt_{i}=\sum_{j}t_{j}\mathcal{M}%
_{ji}\;.  \label{PTG}
\end{equation}%
It is desirable to find a closed formula for the Gram matrix elements
similar to the case treated in \cite{CK07,CKproc}.

\section{Results for boundary fields off the imaginary axis}

In this section we briefly discuss some aspects when $\alpha =\bar{\beta}$
do not lie on the imaginary axis. This case for $\alpha =\bar{\beta}$ on the
unit circle has been investigated previously \cite{HR92,PPM92,AR93} albeit
not in the context of quasi-Hermiticity. Here we relate our discussion to
these previous results.

\subsection{$U_{q}(\mathfrak{gl}(1|1))$ invariance and Hecke algebra}

In order to discuss the quantum group symmetry we consider the Hamiltonian $%
H^{\prime }$ instead of $H$,%
\begin{equation*}
H^{\prime }=\sum_{m=1}^{M}H_{m}^{\prime },\qquad H_{m}^{\prime }=-\frac{%
\sigma _{m}^{x}\sigma _{m+1}^{x}+\sigma _{m}^{y}\sigma _{m+1}^{y}-\alpha
^{-1}~\sigma _{m}^{z}-\alpha ~\sigma _{m+1}^{z}}{2}\;.
\end{equation*}%
This Hamiltonian can be viewed as an element in the Hecke Algebra, compare
with the discussion in \cite{S89,PPM92}.

\begin{definition}
The Hecke algebra $\mathcal{H}_{n}(q)$ with $q\in \mathbb{C}$ is the
associative algebra (over $\mathbb{C}$) obtained from the generators $%
b_{1},...,b_{n-1}$ subject to the relations%
\begin{equation*}
b_{i}b_{i}^{-1}=b_{i}^{-1}b_{i}=1,\qquad
b_{j}b_{j+1}b_{j}=b_{j+1}b_{j}b_{j+1},\qquad b_{i}b_{j}=b_{j}b_{i}\qquad 
\text{for}\qquad |i-j|>1
\end{equation*}%
and%
\begin{equation*}
b_{i}^{2}+(q-q^{-1})b_{i}=1\;.
\end{equation*}
\end{definition}

Setting $n=M$ and $q=-\alpha ^{-1}$ one easily verifies that%
\begin{equation*}
b_{i}\mapsto c_{i}c_{i+1}^{\ast }-c_{i}^{\ast }c_{i+1}-\alpha
^{-1}n_{i}-\alpha (n_{i+1}-1)=\left( 
\begin{array}{cccc}
q & 0 & 0 & 0 \\ 
0 & q-q^{-1} & 1 & 0 \\ 
0 & 1 & 0 & 0 \\ 
0 & 0 & 0 & -q^{-1}%
\end{array}%
\right) _{i,i+1}
\end{equation*}%
and%
\begin{equation*}
b_{i}^{-1}\mapsto c_{i}c_{i+1}^{\ast }-c_{i}^{\ast }c_{i+1}-\alpha
^{-1}(n_{i}-1)-\alpha n_{i+1}=\left( 
\begin{array}{cccc}
-q^{-1} & 0 & 0 & 0 \\ 
0 & 0 & 1 & 0 \\ 
0 & 1 & q^{-1}-q & 0 \\ 
0 & 0 & 0 & q%
\end{array}%
\right) _{i,i+1}
\end{equation*}%
yield a representation of the Hecke algebra $\mathcal{H}_{M}(q=-\alpha
^{-1}) $ over $V^{\otimes M}$ and, furthermore, we have the identity%
\begin{equation}
H^{\prime }=\sum_{i=1}^{M-1}\frac{b_{i}+b_{i}^{-1}}{2}\;.
\end{equation}%
There is a action of the quantum group $\mathcal{U}_{q}(\mathfrak{gl}(1|1))$ 
\cite{S89,Kulish89} which is \textquotedblleft dual\textquotedblright\ to
the action of the Hecke algebra in the sense that both of them commute, i.e.
the Hamiltonian is quantum group invariant.

\begin{definition}
Consider the associative algebra $\mathcal{U}_{q}(\mathfrak{gl}(1|1))$ over $%
\mathbb{C}(q)$ generated by the elements $\{\mathcal{X}^{\pm },\mathcal{Y}%
^{\pm 1},\mathcal{Z}^{\pm 1}\}$ subject to the relations%
\begin{equation}
\mathcal{ZZ}^{-1}=\mathcal{Z}^{-1}\mathcal{Z}=\mathcal{YY}^{-1}=\mathcal{Y}%
^{-1}\mathcal{Y}=1,\mathcal{\quad YX}^{\pm }\mathcal{Y}^{-1}=q^{\pm 1}%
\mathcal{X}^{\pm },\mathcal{\quad \lbrack Y},\mathcal{Z]=[Z},\mathcal{X}%
^{\pm }\mathcal{]}=0
\end{equation}%
and%
\begin{equation}
\lbrack \mathcal{X}^{\pm },\mathcal{X}^{\pm }]_{+}=0,\quad \quad \lbrack 
\mathcal{X}^{+},\mathcal{X}^{-}]_{+}=\frac{\mathcal{Z}-\mathcal{Z}^{-1}}{%
q-q^{-1}}\;.
\end{equation}
\end{definition}

Setting $q=-\alpha ^{-1}$ and identifying%
\begin{equation*}
\mathcal{Y}^{\pm 1}\mapsto q^{\pm S^{z}},\qquad \mathcal{Z}^{\pm 1}\mapsto
q^{\pm M},\qquad \mathcal{X}^{+}\mapsto \sum_{x}q^{\frac{M+1}{2}%
-x}c_{x}^{\ast },\qquad \mathcal{X}^{-}\mapsto \sum_{x}q^{\frac{M+1}{2}%
-x}c_{x}
\end{equation*}%
we obtain a representation of this algebra over $V^{\otimes M}$. Moreover,
its action commutes with the action of the Hecke algebra,%
\begin{equation*}
\lbrack \mathcal{U}_{q}(\mathfrak{gl}(1|1)),\mathcal{H}_{M}(q)]=0,\;
\end{equation*}%
and, thus, with the Hamiltonian $H^{\prime }$. As before the quantum group
invariance is reflected in the existence of a particularly simple solution
for the discrete wave function. Setting as before $q=-\alpha
^{-1}=e^{i\theta }$ the mentioned wave function and creation operator are%
\begin{equation}
\psi _{\theta }(x)=\sqrt{\frac{\sin \theta }{\sin M\theta }}~e^{i(x-\frac{M+1%
}{2})\theta },\qquad \hat{c}_{\theta }^{\ast }=\sqrt{\frac{\sin \theta }{%
\sin M\theta }}~\mathcal{X}^{+}=\sum_{x=1}^{M}\psi _{\theta }(x)c_{x}^{\ast
}\;.  \label{phi}
\end{equation}%
Using these expressions one finds%
\begin{equation}
\lbrack H,\hat{c}_{\theta }^{\ast }]=-2\cos \theta ~\hat{c}_{\theta }^{\ast
}\qquad \text{and}\qquad \lbrack H^{\prime },\hat{c}_{\theta }^{\ast }]=0\ .
\end{equation}%
Similar as in the case $\theta =\pi /2$ and $M$ even the Hamiltonian has
non-trivial Jordan blocks when the norm of the wave function becomes
singular, that is if we choose $\theta $ such that 
\begin{equation*}
\frac{\sin M\theta }{\sin \theta }=0\;.
\end{equation*}

Also the representation theoretic interpretation of $\eta $ extends from $%
\theta =i\pi /2$ to the present, general case: the new inner product $%
\langle \cdot ,\eta \cdot \rangle $ gives rise to a representation of the
quantum group $\mathcal{U}_{q}(\mathfrak{gl}(1|1))$ with the correct $\ast $%
-structure,%
\begin{equation}
\eta \mathcal{X}^{+}=(\mathcal{X}^{-})^{\ast }\eta \;.
\end{equation}%
The above identity has another consequence. Since we also have that $P%
\mathcal{X}^{+}P=(\mathcal{X}^{-})^{\ast }$ it follows that the $C$-operator
obeys 
\begin{equation*}
\lbrack C,\mathcal{U}_{q}(\mathfrak{gl}(1|1))]=0,\qquad C=P\eta
\end{equation*}%
and hence must be an element in the Hecke algebra. This is precisely the
line of reasoning employed in \cite{KW07} to find an algebraic expression
for $C$. However, we need an additional ingredient to complete the analogous
computation: the decomposition of the state space into $\mathcal{U}_{q}(%
\mathfrak{gl}(1|1))$-modules \cite{RS92}. We leave this problem to future
work.

\subsection{Perturbation theory for arbitrary $\protect\theta $}

In this section we wish to confirm that our previous picture for $\theta
=\pi /2$ stays intact for general $\theta $. Namely, we wish to show that
also here the non-Hermitian boundary fields in $H,H^{\prime }$ correspond to
non-local hopping terms in $h$.

For arbitrary $\theta $ we have to modify our previous approach to the
perturbative computation of $\eta =\exp A$. We now define\ 
\begin{equation*}
\alpha =\beta ^{\ast }=-ge^{i\theta },\quad 0\leq g<1:\quad \quad
H_{g}=H_{0}+gH_{1}
\end{equation*}%
with%
\begin{equation}
H_{0}=-\sum_{x=1}^{M-1}\left[ c_{x}^{\ast }c_{x+1}-c_{x}c_{x+1}^{\ast }%
\right] ,\qquad H_{1}=-\frac{e^{i\theta }\sigma _{1}^{z}+e^{-i\theta }\sigma
_{M}^{z}}{2}=-(e^{i\theta }n_{1}+e^{-i\theta }n_{M})\ .
\end{equation}%
For $0<g\ll 1$ we again expand the operator $\eta =\exp A$ using the
Baker-Campbell-Hausdorff formula but now we cannot make the ansatz that $A$
only depends on odd powers of the coupling constant $g$, since $H_{g}^{\ast
}\neq H_{-g}$. Thus all powers of $g$ are occurring in the series expansion,%
\begin{equation}
A=\sum_{n>0}g^{n}A_{n}\ .
\end{equation}%
Collecting once more terms of the same order in $g$ we now arrive at%
\begin{eqnarray}
\lbrack H_{0},A_{1}] &=&2H_{-}  \notag \\
\lbrack H_{0},A_{2}] &=&[A_{1},H_{1}]+\frac{1}{2}%
[A_{1},[A_{1},H_{0}]]=[A_{1},H_{+}]  \notag \\
\lbrack H_{0},A_{3}] &=&[A_{2},H_{+}]-\frac{1}{12}%
[A_{1},[A_{1},[A_{1},H_{0}]]]=[A_{2},H_{+}]+\frac{1}{6}[A_{1},[A_{1},H_{-}]],
\end{eqnarray}%
where%
\begin{equation}
H_{+}=\frac{H_{1}+H_{1}^{\dagger }}{2}=-\cos \theta (n_{1}+n_{M})\quad \text{%
and\quad }H_{-}=\frac{H_{1}-H_{1}^{\dagger }}{2}=i\sin \theta (n_{M}-n_{1})\
.
\end{equation}%
The explicit results for the first two terms in the expansion of the matrix $%
A$ are,%
\begin{eqnarray*}
A_{1} &=&\sum_{x=1}^{M-1}\left[ e^{i\theta }c_{x}^{\ast }c_{x+1}-e^{-i\theta
}c_{x}c_{x+1}^{\ast }\right] , \\
A_{2} &=&\sum_{x=1}^{M-2}\left[ (1+e^{2i\theta })c_{x}^{\ast
}c_{x+2}-(1+e^{-2i\theta })c_{x}c_{x+2}^{\ast }\right]
\end{eqnarray*}%
The corresponding Hermitian Hamiltonian $h$ now reads up to third order in
the coupling $g$,%
\begin{equation*}
h=e^{\frac{A}{2}}He^{-\frac{A}{2}}=H_{0}+\sum_{n=1}^{\infty }g^{n}h_{n},
\end{equation*}%
with%
\begin{eqnarray*}
h_{1} &=&H_{+}=-\cos \theta (n_{1}+n_{M}),\qquad \\
h_{2} &=&\frac{1}{4}~[A_{1},H_{-}]=\frac{\sin \theta }{4i}e^{i\theta
}(c_{1}^{\ast }c_{2}+c_{M-1}^{\ast }c_{M})+\text{h.c.},\qquad \\
h_{3} &=&\frac{1}{4}~[A_{2},H_{-}]=\frac{\sin \theta }{4i}(1+e^{i2\theta
})(c_{1}^{\ast }c_{3}+c_{M-2}^{\ast }c_{M})+\text{h.c.}\;.
\end{eqnarray*}%
Here "h.c." stands for the Hermitian conjugate of the previous term. From
this we infer that for general values of $\theta $ also interactions between
sites separated by an even number are possible.

\section{Conclusions}

In this article we encountered numerous new aspects of the XX spin-chain
with non-Hermitian boundary fields. While this simple integrable model has
been studied intensively before in the literature, it has recently received
renewed attention in connection with logarithmic field theories where the
main focus is on representations of the Temperley-Lieb algebra with
non-trivial Jordan blocks \cite{PRZ06,RS07}.\medskip

The present work has highlighted a very different aspect: purely on physical
grounds one wishes to have a Hermitian quantum Hamiltonian (without Jordan
blocks) in order to ensure a unitary time evolution of the system. This is
one of the essential demands of quantum mechanics. We have seen that this
can be achieved for complex boundary fields with values inside the unit disc
via two different, albeit closely related routes: one can either introduce a
new inner product or perform a similarity transformation to a Hermitian
Hamiltonian. The latter leads to a new physical interpretation of the
non-Hermitian XX spin-chain: it corresponds to a free fermion system with
long range hopping. The probability of long range hopping taking place is
controlled by the absolute value of the complex boundary fields. This new
perspective on the XX spin-chain with non-Hermitian boundary fields entails
a range of other physically interesting questions, such as finite size
effects and correlation functions. We already touched upon the finite size
scaling of the groundstate energy in the text, since the latter contains
information about the respective CFTs in the thermodynamic limit and thus
would connect with the discussion in \cite{Saleur92,PRZ06,RS07}. Due to the
absence of an exact solution for the Bethe roots when the boundary fields
lie within the unit disc, we were unable to obtain conclusive results. An
alternative approach might be to find a field theoretic model which allows
one to compute the partition function, similar as it has been the case for
critical dense polymers on the lattice \cite{Saleur92}.\medskip

To find the correct field theoretic counterpart might be facilitated by the
algebraic structures pointed out in this article. We explicitly constructed
representations of the quantum group $U_{q}(\mathfrak{sl}_{2})$ with $q=%
\sqrt{-1}$ and the superalgebra $U(\mathfrak{gl}(1|1))$ which can be
extended from the unit circle along the imaginary axis. Any field theory
describing the thermodynamic limit should reflect these algebraic features.
Recall that for $M$ odd these algebras provided symmetries, while they did
not for $M$ even. Our careful analysis showed that we could tie the
appearance of (non-trivial) Jordan blocks in the Hamiltonian to the
vanishing of the central element among the generators of $U(\mathfrak{gl}%
(1|1))$. For $M$ even this happens precisely if the boundary fields lie on
the unit circle, but the new result here is that it also happens for $M$ odd
just outside the unit circle. Moreover, we recall from the main text that
for $M$ even one obtains $2\times 2$ Jordan blocks, while for $M$ odd one
has $3\times 3$ blocks. Thus the algebraic picture put forward in \cite{RS07}
for $M$ even needs to be modified for $M$ odd and an interesting problem is
whether one can find also here corresponding logarithmic CFTs.\medskip

The superalgebra $U(\mathfrak{gl}(1|1))$ has also been connected to the
discussion of quasi-Hermiticity. The new inner product preserves the natural 
$\ast $-involution of $U(\mathfrak{gl}(1|1))$, thus providing a
representation theoretic interpretation of our construction. It is this
feature which singles out non-Hermitian quantum integrable systems: due to
their underlying algebraic structures they allow for various interpretations
and \emph{exact} constructions of the quasi-Hermiticity operator. This
connection between abstract mathematical structures and physically motivated
concepts such as $PT$-symmetry is mutually beneficial, it leads to new
physical insight and interesting mathematical questions even for a simple
and well-studied model such as the XX spin-chain.\medskip

\textbf{Acknowledgments}. The author is financially supported by a
University Research Fellowship of the Royal Society. He would like to thank
Ken Brown for references on $U(\mathfrak{gl}(1|1))$ and $U_{q}(\mathfrak{gl}%
(1|1))$, Christina Cobbold for advice, Andreas Fring for comments on
non-Hermitian quantum mechanics at an early stage of this work, Catharina
Stroppel for discussions on the dual canonical basis and Robert Weston for a
previous collaboration.

\end{document}